\documentclass{article}

\usepackage[utf8]{inputenc}
\usepackage{algorithm}
\usepackage{amsmath}
\usepackage{amsfonts}
\usepackage{amssymb}
\usepackage{amsthm}
\usepackage{bm}
\usepackage{color}
\usepackage{cprotect} 
\usepackage{mathtools}
\usepackage{graphicx}
\usepackage{hyperref}
\usepackage[utf8]{inputenc} 
\usepackage{subcaption}
\usepackage{todonotes}
\usepackage{url}
\usepackage[capitalize]{cleveref}
\crefname{equation}{}{}

\usepackage[accepted]{icml2019}

\usepackage{enumitem}

\DeclareMathOperator*{\argmin}{arg\,min}
\DeclareMathOperator*{\argmax}{arg\,max}

\newtheorem{theorem}{Theorem}

\newtheorem{lemma}{Lemma}

\newtheorem{proposition}{Proposition}

\newtheorem{remark}{Remark}

\newcommand{\R}{\mathbb{R}}
\newcommand{\E}{\mathbb{E}}
\newcommand{\grad}{\nabla} 
\newcommand{\pset}{\mathcal{P}} 

\newcommand{\psdgt}{\succ}

\icmltitlerunning{Stein Point Markov Chain Monte Carlo}

\begin{document}

\twocolumn[
\icmltitle{Stein Point Markov Chain Monte Carlo}

\icmlsetsymbol{equal}{*}

\begin{icmlauthorlist}
\icmlauthor{Wilson Ye Chen}{equal,ism}
\icmlauthor{Alessandro Barp}{equal,imp,ati}
\icmlauthor{Fran\c{c}ois-Xavier Briol}{cam,ati}
\icmlauthor{Jackson Gorham}{od}
\icmlauthor{Mark Girolami}{cam,ati}
\icmlauthor{Lester Mackey}{equal,mr}
\icmlauthor{Chris. J. Oates}{ncl,ati}
\end{icmlauthorlist}

\icmlaffiliation{ism}{Institute of Statistical Mathematics}
\icmlaffiliation{imp}{Imperial College London}
\icmlaffiliation{ati}{Alan Turing Institute}
\icmlaffiliation{cam}{University of Cambridge}
\icmlaffiliation{od}{OpenDoor}
\icmlaffiliation{mr}{Microsoft Research}
\icmlaffiliation{ncl}{Newcastle University}

\icmlcorrespondingauthor{Lester Mackey}{lmackey@microsoft.com}
\icmlcorrespondingauthor{Chris. J. Oates}{chris.oates@ncl.ac.uk}

\icmlkeywords{Stein's method, Bayesian computation}

\vskip 0.3in
]

\printAffiliationsAndNotice{\icmlEqualContribution} 

\begin{abstract}
An important task in machine learning and statistics is the approximation of a probability measure by an empirical measure supported on a discrete point set. 
Stein Points are a class of algorithms for this task, which proceed by sequentially minimising a Stein discrepancy between the empirical measure and the target and, hence, require the solution of a non-convex optimisation problem to obtain each new point. 
This paper removes the need to solve this optimisation problem by, instead, selecting each new point based on a Markov chain sample path. 
This significantly reduces the computational cost of Stein Points and leads to a suite of algorithms that are straightforward to implement. 
The new algorithms are illustrated on a set of challenging Bayesian inference problems, and rigorous theoretical guarantees of consistency are established.
\end{abstract} 


\section{Introduction}
\label{sec:introduction}

The task that we consider in this paper is to approximate a Borel probability measure $P$ on an open and convex set $\mathcal{X} \subseteq \mathbb{R}^d$, $d \in \mathbb{N}$, with an empirical measure $\hat{P}$ supported on a discrete point set $\{x_i\}_{i=1}^n \subset \mathcal{X}$. 
To limit scope we restrict attention to uniformly-weighted empirical measures; $\hat{P} = \frac{1}{n} \sum_{i=1}^n \delta_{x_i}$ where $\delta_{x}$ is a Dirac measure on $x$.
The \emph{quantisation} \citep{Graf2007} of $P$ by $\hat{P}$ is an important task in computational statistics and machine learning.
For example, quantisation facilitates the approximation of integrals $\int_{\mathcal X} f \mathrm{d}P$ of measurable functions $f:\mathcal X \rightarrow \mathbb{R}$ using cubature rules $f \mapsto \frac{1}{n} \sum_{i=1}^n f(x_i)$. 
More generally, quantisation underlies a broad spectrum of algorithms for uncertainty quantification that must operate subject to a finite computational budget.
Motivated by applications in Bayesian statistics, our focus is on the situation where $P$ admits a density $p$ with respect to the Lebesgue measure on $\mathcal{X}$ but this density can only be evaluated up to an (unknown) normalisation constant.
Specifically, we assume that $p = \frac{\tilde{p}}{C}$ where $\tilde{p}$ is an un-normalised density and $C>0$, such that both $\tilde{p}$ and $\nabla \log \tilde{p}$, where $\nabla = (\frac{\partial}{\partial x^1},\ldots,\frac{\partial}{\partial x^d})$, can be (pointwise) evaluated at finite computational cost. 

A popular approach to this task is Markov chain Monte Carlo \citep[MCMC;][]{Robert2004}, where the sample path of an ergodic Markov chain with invariant distribution $P$ constitutes a point set $\{x_i\}_{i=1}^n$.
MCMC algorithms exploit a range of techniques to construct Markov transition kernels which leave $P$ invariant, based (in general) on pointwise evaluation of $\tilde{p}$ \citep{Metropolis1953} or (sometimes) on pointwise evaluation of $\nabla \log \tilde{p}$ and higher-order derivative information \citep{Girolami2011}.
In a favourable situation, the MCMC output will be approximately independent draws from $P$.
However, in this case the $\{x_i\}_{i=1}^n$ will typically not be a \emph{low discrepancy} point set \citep{Dick2010} and as such the quantisation of $P$ performed by MCMC will be sub-optimal.
In recent years several attempts have been made to deveop improved algorithms for quantisation in the Bayesian statistical context as an alternative to MCMC:
\begin{itemize}
\item \textbf{Minimum Energy Designs} (MED)
In \cite{RoshanJoseph2015,Joseph2017} it was proposed to obtain a point set $\{x_i\}_{i=1}^n$ by using a numerical optimisation method to approximately minimise an energy functional $\mathcal{E}_{\tilde{p}}(\{x_i\}_{i=1}^n)$ that depends on $P$ only through $\tilde{p}$ rather than through $p$ itself.
Though appealing in its simplicity, MED has yet to receive a theoretical treatment that accounts for the imperfect performance of the numerical optimisation method.

\item \textbf{Support Points} The method of \cite{Mak2018} first generates a large MCMC output $\{\tilde{x}_i\}_{i=1}^N$ and from this a subset $\{x_i\}_{i=1}^n$ is selected in such a way that a low-discrepancy point set is obtained.
(This can be contrasted with classical \emph{thinning} in which an arithmetic subsequence of the MCMC output is selected.)
At present, a theoretical analysis that accounts for the possible poor performance of the MCMC method has not yet been announced.

\item \textbf{Transport Maps and QMC} The method of \cite{Parno2015} aims to learn a \emph{transport map} $T : \mathcal{X} \rightarrow \mathcal{X}$ such that the pushforward measure $T_\# Q$ corresponds to $P$, where $Q$ is a distribution for which quantisation by a point set $\{\tilde{x}_i\}_{i=1}^n$ is easily performed, for instance using quasi-Monte Carlo (QMC) \citep{Dick2010}.
Then quantisation of $P$ is provided by the point set $\{T(\tilde{x}_i)\}_{i=1}^n$.
The flexibility in the construction of a transport map allows several algorithms to be envisaged, but an end-to-end theoretical treatment is not available at present.

\item \textbf{Stein Variational Gradient Descent} (SVGD)
A popular methodology due to \citep{Liu2016SVGD} aims to take an arbitrary initial point set $\{x_i^{0}\}_{i=1}^n$ and to construct a discrete time dynamical system $x_i^{t} = g_{\tilde{p}}(x_1^{t-1},\dots,x_n^{t-1})$, indexed by time $t$ and dependent on $\tilde{p}$, such that $\lim_{t \rightarrow \infty} \{x_i^{t}\}_{i=1}^n$ provides a quantisation of $P$.
This can be viewed as a discretisation of a particular gradient flow that has $P$ as a fixed point \citep{Liu2017GradientFlow}.
However, a generally applicable theoretical analysis of the SVGD method itself is not available \citep[note that a compactness assumption on $\mathcal{X}$ was required in][]{Liu2017GradientFlow}.
Note also that, unlike the other methods discussed in this section, SVGD does not readily admit an \emph{extensible} construction; that is, the number $n$ of points must be {\it a priori} fixed.

\item \textbf{Stein Points} (SP)
The authors of \cite{Chen2018SteinPoints} proposed to select a point set $\{x_i\}_{i=1}^n$ that approximately minimises a \emph{kernel Stein discrepancy} \citep[KSD;][]{Liu2016,Chwialkowski2016,Gorham2017} between the empirical measure and the target $P$.
The KSD can be exactly computed with a finite number of pointwise evaluations of $\nabla \log \tilde{p}$ and, for the (non-convex) minimisation, a variety of numerical optimisation methods can be applied.
In contrast to the other methods just discussed, SP does admit a end-to-end theoretical treatment when a grid search procedure is used as the numerical optimisation method \citep[Thms. 1 \& 2 in][]{Chen2018SteinPoints}.
\end{itemize}

An empirical comparison of several of the above methods on a selection of problems arising in computational statistics was presented in \cite{Chen2018SteinPoints}.
The conclusion of that work was that MED and SP provided broadly similar performance-per-computational-cost at the quantisation task, where the performance was measured by the Wasserstein distance to the target and the computational cost was measured by the total number of evaluations of either $\tilde{p}$ or its gradient.
In some situations, SVGD provided superior quantisation to MED and SP but this was achieved at a substantially higher computational cost.
At the same time, it was observed that \emph{all} algorithms considered provided improved quantisation compared to MCMC, but at a computational cost that was substantially higher than the corresponding cost of MCMC.

\begin{figure}[t!]
\centering
\includegraphics[width=0.45\textwidth,clip,trim = 1.1cm 1cm 1cm 0cm]{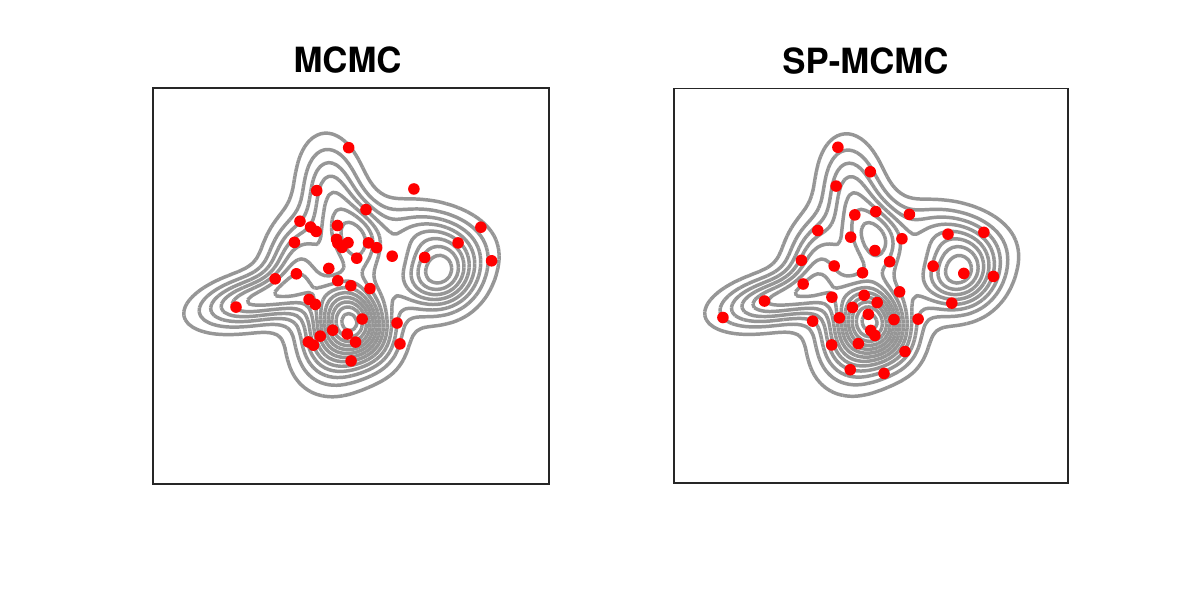} 
\caption{Illustration of Monte Carlo points (MC; left) and Stein Point Markov chain Monte Carlo (SP-MCMC; right) on a Gaussian mixture target $P$. SP-MCMC provides better space-filling properties than MC.}
\label{fig:illustration_stein_points}
\end{figure}

In this paper, we propose Stein Point Markov chain Monte Carlo (SP-MCMC), aiming to provide strong performance at the quantisation task (see Fig. \ref{fig:illustration_stein_points}) but at substantially reduced computational cost compared to the original SP method. 
Our contributions are summarised as follows:
\begin{itemize}
\item The global optimisation subroutine in SP, whose computational cost was exponential in dimension $d$, is replaced by a form of local search based on MCMC. 
This allows us to make use of efficient transition kernels for exploration of $\mathcal{X}$, which in turn improves performance in higher dimensions and reduces the overall computational cost.

\item Our construction requires a new Markov chain to be initialised each time a point $x_n$ is added, however the initial distribution of the chain does not need to coincide with $P$.
This enables us to develop an efficient criterion for initialisation of the Markov chains, based on the introduced notion of the ``most influential'' point in $\{x_i\}_{i=1}^{n-1}$, as quantified by KSD. 
This turns our sequence of local searches into a global-like search, and also leads to automatic ``mode hopping'' behaviour when $P$ is a multi-modal target.

\item The consistency of SP-MCMC is established under a $V$-uniform ergodicity condition on the Markov kernel.

\item SP-MCMC is shown, empirically, to outperform MCMC, MED, SVGD and SP when applied to posterior computation in the Bayesian statistical context.
\end{itemize}

The paper is structured as follows:
In Section \ref{sec:background} we review the central notions of Stein's method and KSD, as well as recalling the original SP method.
The novel methodology is presented in Section \ref{sec:methodology}.
This is assessed experimentally in Section \ref{sec:experiments} and theoretically in Section \ref{sec:theory}.
Conclusions are drawn in Section \ref{sec:conclusion}.


\section{Background}
\label{sec:background}
In Section \ref{subsec: discrepancy} we recall the construction of KSD, then in Section \ref{sec: Stein points} the SP method of \cite{Chen2018SteinPoints}, which is based on minimisation of KSD, is discussed.

\subsection{Discrepancy and Stein's Method} \label{subsec: discrepancy}

A \emph{discrepancy} is a notion of how well an empirical measure, based on a point set $\{x_i\}_{i=1}^n \subset \mathcal{X}$, approximates a target $P$. 
One popular form of discrepancy is the \textit{integral probability metric} (IPM) \citep{Muller1997}, which is based on a set $\mathcal{F}$ consisting of functionals on $\mathcal{X}$, and is defined as:
\begin{equation}\label{eq:IPMs}
D_{\mathcal{F},P}(\{x_i\}_{i=1}^n)
\; := \;
\textstyle \sup_{f \in \mathcal{F}} \left| \frac{1}{n} \sum_{i=1}^n f(x_i) - \int_{\mathcal X} f \mathrm{d}P \right| 
\end{equation}
The set $\mathcal{F}$ is required to be measure-determining in order for the IPM to be a genuine metric.
Certain sets $\mathcal{F}$ lead to familiar notions, such as the Wasserstein distance, but direct computation of an IPM will generically require exact integration against $P$; a demand that is not met in the Bayesian context.
In order to construct an IPM that \emph{can} be computed in the Bayesian context, \cite{Gorham2015} proposed the notion of a \emph{Stein discrepancy}, based on Stein's method \citep{Stein1972}.
This consists of finding an operator $\mathcal{A}$, called a \emph{Stein operator}, and a function class $\mathcal{G}$, called a \emph{Stein class}, which satisfy the \emph{Stein identity} $ \int_{\mathcal{X}} \mathcal{A}g \mathrm{d}P = 0$ for all $g \in \mathcal{G}$. 
Taking $\mathcal{F} = \mathcal{A} \mathcal{G}$ to be the image of $\mathcal{G}$ under $\mathcal{A}$ in \cref{eq:IPMs} leads directly to the \textit{Stein discrepancy}:
\begin{equation}\label{eq:stein_discrepancy}
D_{\mathcal{AG},P}(\{x_i\}_{i=1}^n)
\; = \;
\textstyle \sup_{g \in \mathcal{G}} \left| \frac{1}{n} \sum_{i=1}^n \mathcal{A}g(x_i) \right| 
\end{equation}
A particular choice of $\mathcal{A}$ and $\mathcal{G}$ was studied in \cite{Gorham2015} with the property that exact computation can be performed based only on point-wise evaluation of $\nabla \log \tilde{p}$.
The computation of this \emph{graph Stein discrepancy} reduced to solving $d$ independent linear programs in parallel with $O(n)$ variables and constraints.

To eliminate the the reliance on a linear program solver, \cite{Liu2016,Chwialkowski2016,Gorham2017} proposed \emph{kernel Stein discrepancies}, alternative Stein discrepancies \cref{eq:stein_discrepancy} with embarrassingly parallel, closed-form values.
For the remainder we assume that $p > 0$ on $\mathcal{X}$.
The canonical KSD is obtained by taking the Stein operator $\mathcal{A}$ to be the Langevin operator $\mathcal{A} g :=  \frac{1}{\tilde{p}} \nabla \cdot (\tilde{p} g)$ and the Stein class $\mathcal{G} = B(\mathcal{K}^d)$ to be the unit ball of a space of vector-valued functions, formed as a $d$-dimensional Cartesian product of scalar-valued reproducing kernel Hilbert spaces $\mathcal{K}$ (RKHS) \citep{Berlinet2004}. 
(Throughout we use $\nabla \cdot$ to denote divergence and $\langle \cdot, \cdot \rangle$ to denote the Euclidean inner product.)
Recall that an RKHS $\mathcal{K}$ is a Hilbert space of functions with inner product $\langle\cdot,\cdot\rangle_{k}$ and induced norm $\|\cdot\|_k$, and there is a function $k: \mathcal X \times \mathcal X \rightarrow \mathbb{R}$, called a \emph{kernel}, such that $\forall x \in\mathcal  X$, we can write the evaluation functional $f(x) = \langle f, k(\cdot,x)\rangle_{k}$ $\forall f \in \mathcal{K}$.
It is assumed that the mixed derivatives $\partial^2 k(x,y)/\partial x^i \partial y^{j}$ and all lower-order derivatives are continuous and uniformly bounded. 
For $\mathcal{X}$ bounded, with piecewise smooth boundary denoted $\partial \mathcal{X}$, outward normal denoted $\mathrm{n}$ and surface element denoted $\mathrm{d}\sigma(x)$, the conditions $\oint_{\partial \mathcal{X}} k(x,x') p(x) \mathrm{n}(x) \mathrm{d}\sigma(x') = 0$, $\textstyle \oint_{\partial \mathcal{X}} \nabla_x k(x,x') \cdot \mathrm{n}(x) p(x) \mathrm{d}\sigma(x') = 0$ are sufficient for the Stein identity to hold; c.f. Lemma 1 in \cite{Oates2017}.
For $\mathcal{X} = \mathbb{R}^d$, a sufficient condition is $\textstyle \int_{\mathcal{X}} \| \nabla \log p(x) \|_2 \mathrm{d}P(x) < \infty$; c.f. Prop. 1 of \cite{Gorham2017}.
The image $\mathcal{A}\mathcal{G} = B(\mathcal{K}_0)$ is the unit ball of another RKHS, denoted $\mathcal{K}_0$, whose kernel is \citep{Oates2017}:
\begin{align}
k_0(x,x') & \; = \; \nabla_x \cdot \nabla_{x'} k(x,x') + \left\langle \nabla_x k(x,x') , \nabla_{x'} \log \tilde{p}(x') \right\rangle \notag\\ 
 & \quad + \left\langle \nabla_{x'} k(x,x') , \nabla_x \log \tilde{p}(x) \right\rangle \notag\\
 & \quad + k(x,x') \left\langle \nabla_{x} \log \tilde{p}(x) , \nabla_{x'} \log \tilde{p}(x') \right\rangle \label{eq:stein_kernel}
\end{align}
In this case, \cref{eq:stein_discrepancy} corresponds to a maximum mean discrepancy \citep[MMD;][]{Gretton2006} in the RKHS $\mathcal{K}_0$ and thus can be explicitly computed.
The Stein identity implies that $\int_{\mathcal{X}} k_0(x,\cdot) \mathrm{d}P \equiv 0$.
Thus we denote the KSD between the empirical measure $\frac{1}{n} \sum_{i=1}^n \delta_{x_i}$ and the target $P$ (in a small abuse of notation) as
\begin{equation}
D_{\mathcal{K}_0,P}(\{x_i\}_{i=1}^n) \; := \; \textstyle \sqrt{\frac{1}{n^2}\sum_{i,j = 1}^n k_0(x_i,x_j)}.
\end{equation}
Under regularity assumptions \citep{Gorham2017,Chen2018SteinPoints,Huggins2018}, the KSD controls classical weak convergence of the empirical measure to the target. 
This motivates selecting the $\{x_i\}_{i=1}^n$ to minimise the KSD, and to this end we now recall the SP method of \citep{Chen2018SteinPoints}.

\subsection{Stein Points} \label{sec: Stein points}

The \textit{Stein Point} (SP) method due to \citep{Chen2018SteinPoints} selects points $\{x_i\}_{i=1}^n$ to approximately minimise $D_{\mathcal{K}_0,P}\left(\{x_i\}_{i=1}^n \right)$. 
This is of course a challenging non-convex and multivariate problem in general. 
For this reason, two sequential strategies were proposed.
The first, called \textit{Greedy} SP, was based on greedy minimisation of KSD, whilst the second, called \textit{Herding} SP, was based on Frank-Wolfe minimisation of KSD. 
In each case, at iteration $j \in \{1,\dots,n\}$ of the algorithm, the points $\{x_i\}_{i=1}^{j-1}$ have been selected and a global search method is used to select the next point $x_j \in \mathcal{X}$. 
To limit scope we restrict the discussion below to Greedy SP, as this has stronger theoretical guarantees and has been shown empirically to outperform Herding SP.
The convergence of Greedy SP was established in Theorem 2 of \citep{Chen2018SteinPoints} when $k_0$ is a $P$-sub-exponential kernel (Def. 1 of Chen et al.). 
More precisely, assume that for some pre-specified tolerance $\delta > 0$, the resulting point sequence satisfies the following identity $\forall j \in \{1,\ldots,n\}:$
\begin{eqnarray*}
D_{\mathcal{K}_0,P}(\{x_i\}_{i=1}^j)^2 \leq {\textstyle \frac{\delta}{j^2} } + \inf_{x \in \mathcal{X}} D_{\mathcal{K}_0,P}(\{x_i\}_{i=1}^{j-1} \cup \{x\})^2.
\end{eqnarray*}
Then it was shown that $\exists \; c_1,c_2>0$ such that
\begin{align}
D_{\mathcal{K}_0,P}(\{x_i\}_{i=1}^n) &\leq e^{\pi/2} \textstyle \sqrt{\frac{2 \log(n)}{c_2 n} + \frac{c_1}{n} + \frac{\delta}{n}} \label{eq: SP theory result}
\end{align}
so that KSD is asymptotically minimised.
However, a significant limitation of the SP method is that it requires a \textit{global} (non-convex) minimisation problem over $\mathcal{X}$ to be (approximately) solved in order to select the next point. 
In practice, the global search at iteration $j$ can be facilitated by a grid search over $\mathcal{X}$, but this procedure entails a computational cost that is exponential in the dimension $d$ of $\mathcal{X}$ and even in modest dimension this becomes impractical. 

The main contribution of the present paper is to re-visit the SP method and to study its behaviour when the global search is replaced with a \textit{local} search, facilitated by a MCMC method.
To proceed, two main challenges must be addressed:
First, an appropriate local optimisation procedure must be developed.
Second, the theoretical convergence of the modified algorithm must be established.
In the next section we address the first challenge by presenting our novel methodological development.


\section{Methodology}
\label{sec:methodology}
In Section \ref{subset: MCMC search} we present the novel SP-MCMC method.
Then in Section \ref{subsec: precondition} we describe how the kernel $k$ can be pre-conditioned to improve performance in SP-MCMC.

\subsection{SP-MCMC} \label{subset: MCMC search}

In this paper, we propose to replace the global minimisation at iteration $j$ of the SP method of \cite{Chen2018SteinPoints} with a local search based on a $P$-invariant Markov chain of length $m_j$, where the sequence $(m_j)_{j \in \mathbb{N}}$ is to be specified. 
The proposed SP-MCMC method proceeds as follows:
\begin{enumerate}
\item Fix an initial point $x_1 \in \mathcal{X}$. 
\item For $j = 2,\dots,n$:
\begin{enumerate}
\item [i.] Select an index $i^* \in \{1,\dots,j-1\}$ according to some criterion \verb!crit!$(\{x_i\}_{i=1}^{j-1})$, to be defined.
\item [ii.] Run a $P$-invariant Markov chain, initialised at $x_{i^*}$, for $m_j$ iterations and denote the realised sample path as $(y_{j,l})_{l=1}^{m_j}$.
\item [iii.] Set $x_j = y_{j,l}$ where $l \in \{1,\dots,m_j\}$ minimises $D_{\mathcal{K}_0,P}(\{x_i\}_{i=1}^{j-1} \cup \{y_{j,l}\})$.
\end{enumerate}
\end{enumerate}

It remains to specify the sequence $(m_j)_{j \in \mathbb{N}}$ and the criterion \verb!crit!. Precise statements about the effect of these choices on convergence are reserved for the theoretical treatment in Section \ref{sec:theory}.
For the criterion \verb!crit!, three different approaches are considered:
\begin{itemize}
\item \verb!LAST! selects the point last added: $i^* := j-1$.
\item \verb!RAND! selects $i^*$ uniformly at random in $\{1,\dots,j-1\}$.
\item \verb!INFL! selects $i^*$ to be the index of a most influential point in $\{x_i\}_{i=1}^{j-1}$. 
Specifically, we call $x_{i^*}$ a \textit{most influential} point if removing it from our point set creates the greatest increase in KSD. i.e. $i^*$ maximises $D_{\mathcal{K}_0,P}(\{x_i\}_{i=1}^{j-1} \setminus \{x_{i^*}\})$.
\end{itemize}

SP-MCMC overcomes the main limitation facing the original SP method; the global search is avoided. 
Indeed, the cost of simulating $m_j$ steps of a $P$-invariant Markov chain will typically be just a fraction of the cost of implementing a global search method. 
The number of iterations $m_j$ acts as a lever to trade-off approximation quality against computational cost, with larger $m_j$ leading on average to an empirical measure with lower KSD.
The precise relationship is elucidated in Section \ref{sec:theory}.

\begin{remark}[KSD has low overhead]
A large number of modern MCMC methods, such as the Metropolis-adjusted Langevin algorithm (MALA) and Hamiltonian Monte Carlo, exploit evaluations of $\nabla \log \tilde{p}$ to construct a $P$-invariant Markov transition kernel \citep{Barp2018HMCReview}.
If such an MCMC method is used, the gradient information $\nabla \log \tilde{p}(x_{i^*})$ is computed during the course of the MCMC and can be recycled in the subsequent computation of KSD.
\end{remark}

\begin{remark}[Automatic mode-hopping]
Although the Markov chain is used only for a local search, the \emph{initialisation} criteria \verb!RAND! and \verb!INFL! offer the opportunity to jump to any point in the set $\{x_i\}_{i=1}^{j-1}$ and thus can facilitate global exploration of the state space $\mathcal{X}$. 
The \verb!INFL! criteria, in particular, favours areas of $\mathcal{X}$ that are under-represented in the point set and thus, for a multi-modal target $P$, one can expect ``mode hopping'' from near an over-represented mode to near an under-represented mode of $P$.
\end{remark}

\begin{remark}[Removal of bad points] \label{rem: bad points}
A natural extension of the SP-MCMC method allows for the possibility of removing a ``bad'' point from the current point set.
That is, at iteration $j$ we may decide, according to some probabilistic or deterministic schedule, to remove a point $x_{i^*}$ that minimises $D_{\mathcal{K}_0,P}(\{x_i\}_{i=1}^{j-1} \setminus \{x_{i^*}\})$.
This extension was also investigated and results are reserved for Section \ref{subsec: bad points supp}.
\end{remark}

\begin{remark}[Sequence vs set]
If the number $n$ of points is pre-specified, then after the $n$ point is selected one can attempt to further improve the point set by applying (e.g.) co-ordinate descent to the KSD interpreted as a function $D_{\mathcal{K}_0,P} : \mathcal{X}^n \rightarrow [0,\infty)$; see \cite{Chen2018SteinPoints}. 
To limit scope, this was not considered.
\end{remark}

\subsection{Pre-conditioned Kernels for SP-MCMC}
\label{subsec: precondition}

The original analysis of \cite{Gorham2017} focussed on the inverse multiquadric (IMQ) kernel $k(x,x') = \left( 1 + \lambda^{-2} \|x - x'\|_2^2 \right)^\beta$ for some length-scale parameter $\lambda > 0$ and exponent $\beta \in (-1,0)$; alternative kernels were considered in \cite{Chen2018SteinPoints}, but the IMQ kernel was observed to lead to the best empirical approximations as quantified objectively by the Wasserstein distance between the empirical measure and the target.
Thus, in this paper we focus on the IMQ kernel.
However, in order to improve the performance of the algorithm, we propose to allow for \textit{pre-conditioning} of the kernel; that is, we consider
\begin{align}
k(x,x') &= \big( 1 + \|\Lambda^{-\frac{1}{2}}(x - x')\|_2^2 \big)^\beta \label{eq: IMQ kernel}
\end{align}
for some symmetric positive definite matrix $\Lambda$.
The use of pre-conditioned kernels was recently proposed in the context of SVGD in \cite{Detommaso2018}, where $\Lambda^{-1}$ was taken to be an approximation to the expected Hessian $- \int \nabla_x \nabla_x^\top \log \tilde{p}(x) \mathrm{d}P(x)$ of the negative log target.
Note that the matrix $\Lambda$ can also form part of a MCMC transition kernel, such as the pre-conditioner matrix in MALA \citep{Girolami2011}.
Sufficient conditions for when a pre-conditioned kernel ensures that KSD controls classical weak convergence of the empirical measure to the target are established in Section \ref{sec:theory}.


\section{Experimental Results}
\label{sec:experiments}
In this section our attention turns to the empirical performance of SP-MCMC.
The experimental protocol is explained in Section \ref{subsec: protocol} and specific experiments are described in Sections \ref{subsec: GMM}, \ref{subsec: IGARCH} and \ref{subsec: ODE}.

\subsection{Experimental Protocol} \label{subsec: protocol}

To limit scope, we present a comparison of SP-MCMC to the original SP method, as well as to MCMC, MED and SVGD.
All experiments involving SP-MCMC, SP or SVGD in this paper were based on the IMQ kernel in \cref{eq: IMQ kernel} with $\beta = -\frac{1}{2}$.
The preconditioner matrix $\Lambda$ was taken either to be a sample-based approximation to the covariance matrix of $P$ (Secs. \ref{subsec: GMM} and \ref{subsec: IGARCH}), generated by running a short MCMC, or $\Lambda \propto I$ (Sec. \ref{subsec: ODE}); however, in each experiment $\Lambda$ was fixed across all methods being compared. 
The Markov chains used for SP-MCMC and MCMC in this work employed either a random walk Metropolis (RWM) or a MALA transition kernel, described in Appendix \ref{ap: mala explain}.
Our implementations of MED and SVGD are described in Appendix \ref{ap: benchmark methods}.

Three experiments of increasing sophistication were considered.\footnote{Code to reproduce all experiments can be downloaded at \url{https://github.com/wilson-ye-chen/sp-mcmc}. }
First, in Section \ref{subsec: GMM} we consider a simple Gaussian mixture target in order to explore SP-MCMC and investigate sensitivity to the degrees of freedom in this new method.
Second, in Section \ref{subsec: IGARCH} we revisit one of the experiments in \cite{Chen2018SteinPoints}, in order to directly compare against SP, MCMC, MED and SVGD.
Third, in Section \ref{subsec: ODE} we consider a more challenging application to Bayesian parameter inference in an ordinary differential equation (ODE) model.

\subsection{Gaussian Mixture Model} \label{subsec: GMM}

\begin{figure}[t!]
\begin{center}
\includegraphics[width = \columnwidth,clip,trim = 1.3cm 1.5cm 2.1cm 1cm]{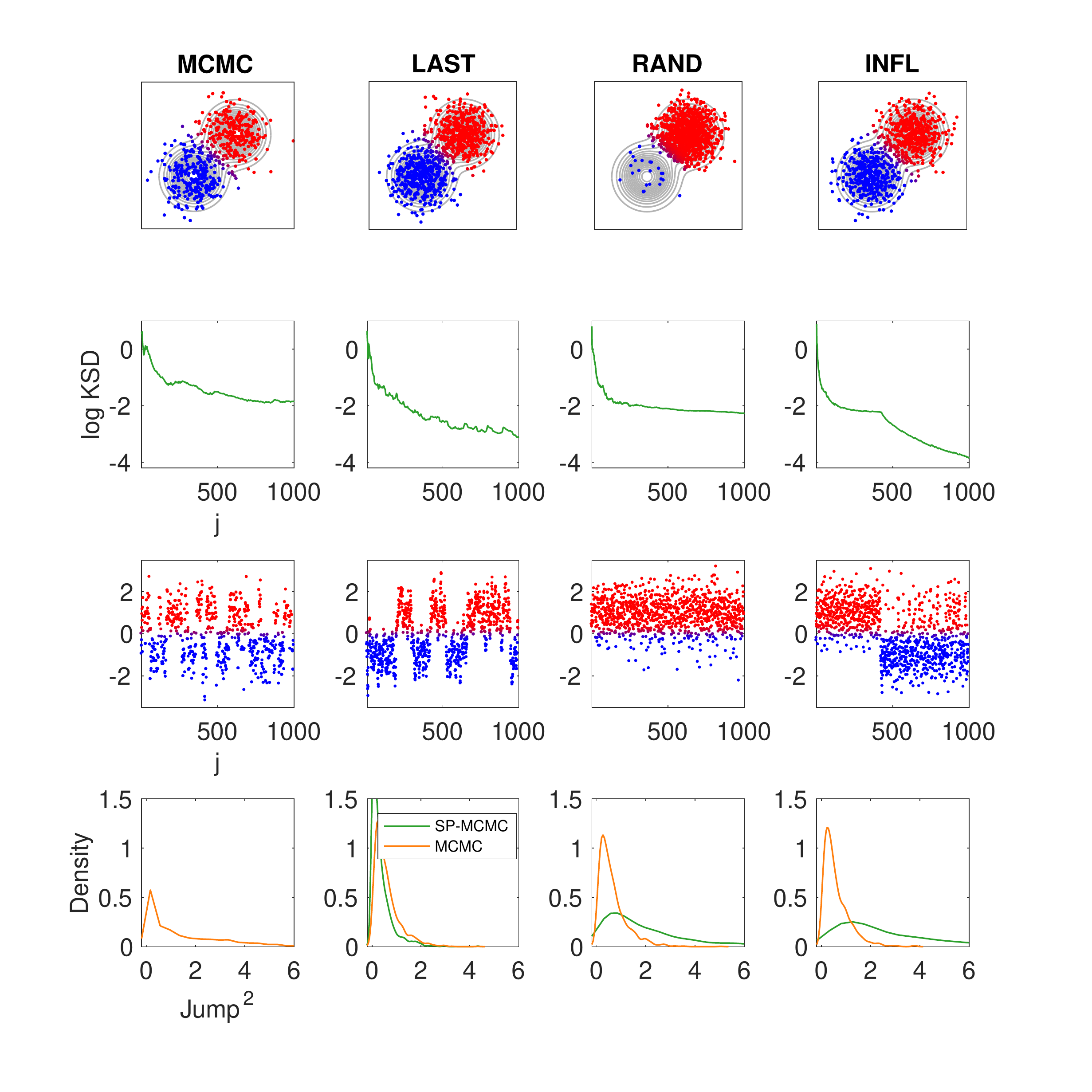}
\cprotect\caption{Gaussian mixture experiment in dimension $d = 2$.
Columns (left to right): MCMC, SP-MCMC with \verb!LAST!, SP-MCMC with \verb!RAND!, SP-MCMC with \verb!INFL!.
Top row: Point sets of size $n = 1000$ produced by MCMC and SP-MCMC.
(Point colour indicates the mode to which they are closest.)
Second row: Trace plot of $\log D_{\mathcal{K}_0,P}(\{x_i\}_{i=1}^j)$ as $j$ is varied from $1$ to $n$.
Third row: Trace plots of the sequence $(x_i)_{i=1}^n$, projected onto the first coordinate.
Bottom row: Distribution of the squared jump distance $\|x_j - x_{j-1}\|_2^2$ (green) compared to the quantities $\|y_{j,m_j} - y_{j,1}\|_2^2$ associated with the Markov chains (orange) used during the course of each method.
}
\label{fig:Gauss mixture}
\end{center}
\end{figure}

For exposition we let $\sigma^2 = 0.5$ and consider a $d =2$ dimensional Gaussian mixture model $P = \frac{1}{2}\mathcal{N}(-\mathrm{1},\sigma^2 \mathrm{I}_{d \times d}) + \frac{1}{2}\mathcal{N}(\mathrm{1},\sigma^2 \mathrm{I}_{d \times d})$ with modes at $\mathrm{1} = [1,1]$ and $-\mathrm{1}$.
The performance of MCMC was compared to SP-MCMC for each of the criteria \verb!LAST!, \verb!RAND!, \verb!INFL!.
Note that in this section we do not address computational cost; this is examined in Secs. \ref{subsec: IGARCH} and \ref{subsec: ODE}.
For SP-MCMC the sequence $(m_j)_{j \in \mathbb{N}}$ was set as $m_j = 5$.
Results are presented in Fig. \ref{fig:Gauss mixture} with $n = 1000$.

The point sets produced by SP-MCMC with \verb!LAST! and \verb!INFL! (top row) were observed to provide a better quantisation of the target $P$ compared to MCMC, as captured by the KSD of the empirical measure to the target (second row).
\verb!RAND! did not distribute points evenly between modes and, as a result, KSD was observed to plateau in the range of $n$ displayed.
For MCMC, the proposal step-size $h > 0$ was optimised according to the recommendations in \cite{Roberts2001}, but nevertheless the chain was observed to jump between the two components of $P$ only infrequently (third row, colour-coded).
In contrast, after an initial period where both modes are populated, SP-MCMC under the \verb!INFL! criteria was seen to frequently jump between components of $P$.
Finally, we note that under \verb!INFL! the typical squared jump distance $\|x_j - x_{j-1}\|_2^2$ was greater than the analogous quantities $\|y_{j,m_j} - y_{j,1}\|_2^2$ for the underlying Markov chains that were used (bottom row), despite the latter being optimised according to the recommendations of \cite{Roberts2001}, which supports the view that more frequent mode-hopping is a property of the \verb!INFL! method.
Based on the findings of this experiment, we focus only on \verb!LAST! and \verb!INFL! in the sequel.
The extension where ``bad'' points are removed, described in Remark \ref{rem: bad points}, was explored in supplemental Section \ref{subsec: bad points supp}.

\subsection{IGARCH Model} \label{subsec: IGARCH}

Next our attention turns to whether SP-MCMC improves over the original SP method and how it compares to existing methods such as MED and SVGD when computational cost is taken into account.
To this end we consider an identical experiment to \cite{Chen2018SteinPoints}, based on Bayesian inference for a classical integrated generalised autoregressive conditional heteroskedasticity (IGARCH) model.
The IGARCH model  \citep{Taylor2011}
\begin{eqnarray*}
y_t & = & \sigma_t \epsilon_t, \hspace{40pt} \epsilon_t \stackrel{\text{i.i.d.}}{\sim} \mathcal{N}(0,1) \\
\sigma_t^2 & = & \theta_1 + \theta_2 y_{t-1}^2 + (1-\theta_2) \sigma_{t-1}^2
\end{eqnarray*}
describes a financial time series $(y_t)$ with time-varying volatility $(\sigma_t)$.
The model is parametrised by $\theta = (\theta_1,\theta_2)$, $\theta_1 > 0$ and $0 < \theta_2 < 1$ and Bayesian inference for $\theta$ is considered, based on data $y = (y_t)$ that represent 2,000 daily percentage returns of the S\&P 500 stock index (from December 6, 2005 to November 14, 2013).
Following \cite{Chen2018SteinPoints}, an improper uniform prior was placed on $\theta$. The domain $\mathcal{X} = \mathbb{R}_+ \times (0,1)$ is bounded and, for this example, the posterior $P$ places negligible mass near the boundary $\partial \mathcal{X}$.
This ensures that the boundary conditions described in Sec. \ref{subsec: discrepancy} hold essentially to machine precision, as argued in \cite{Chen2018SteinPoints}.

\begin{figure}
\centering
\includegraphics[width = 0.5\textwidth]{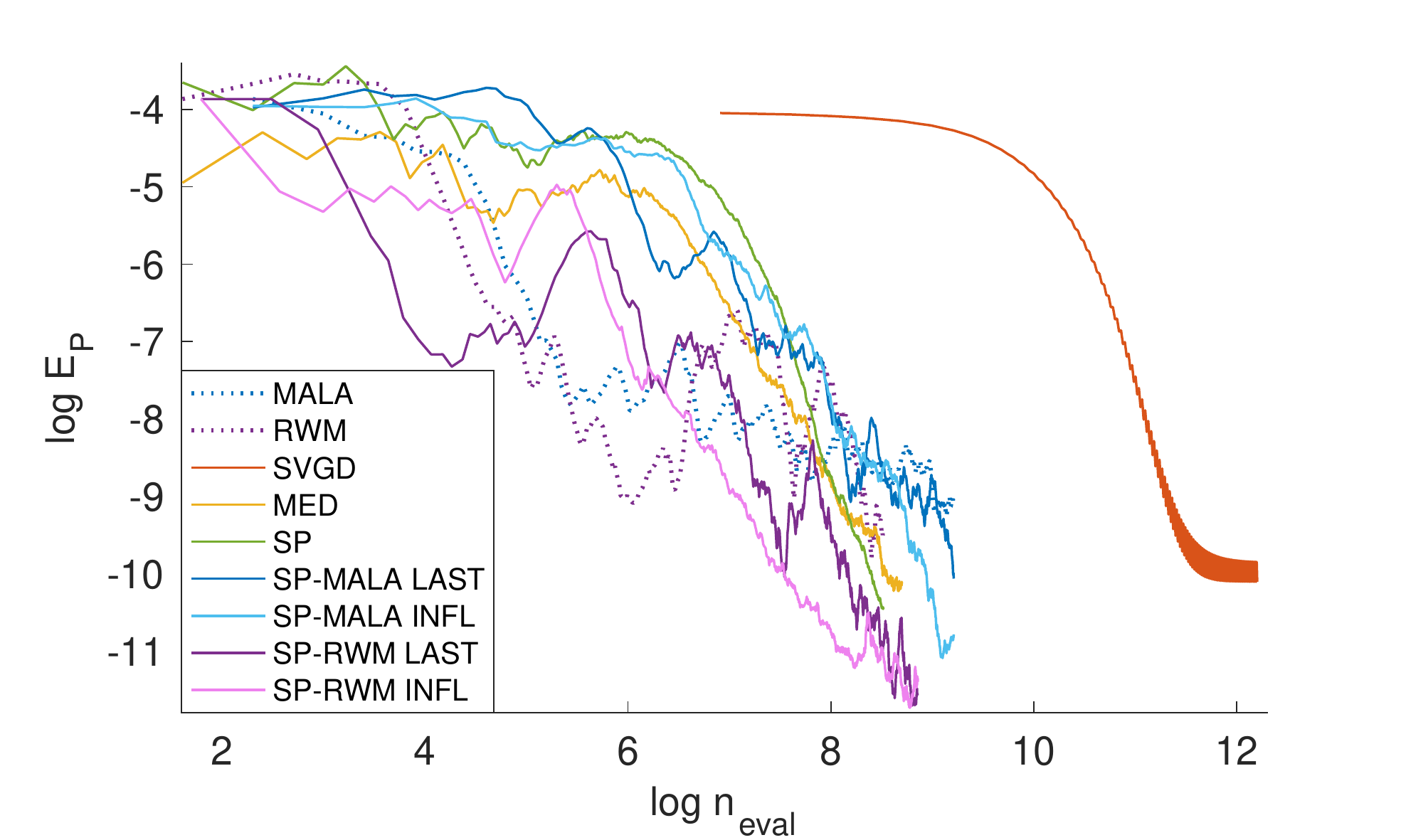}
\caption{IGARCH experiment.
The new SP-MCMC method was compared against the original SP method of \cite{Chen2018SteinPoints}, as well as against MCMC, MED \citep{RoshanJoseph2015} and SVGD \citep{Liu2016SVGD}.
The implementation of all existing methods is described in Appendix \ref{ap: extra results}.
Each method produced an empirical measure $\frac{1}{n}\sum_{i=1}^n \delta_{x_i}$ whose distance to the target $P$ was quantified by the energy distance $E_P$.
The computational cost was quantified by the number $n_{\text{eval}}$ of times either $\tilde{p}$ or its gradient were evaluated.
}
\label{fig:IGARCH}
\end{figure}

\begin{figure*}[t]
\centering
\begin{subfigure}[t]{0.49\textwidth}
\centering
\includegraphics[width = 1.0\textwidth]{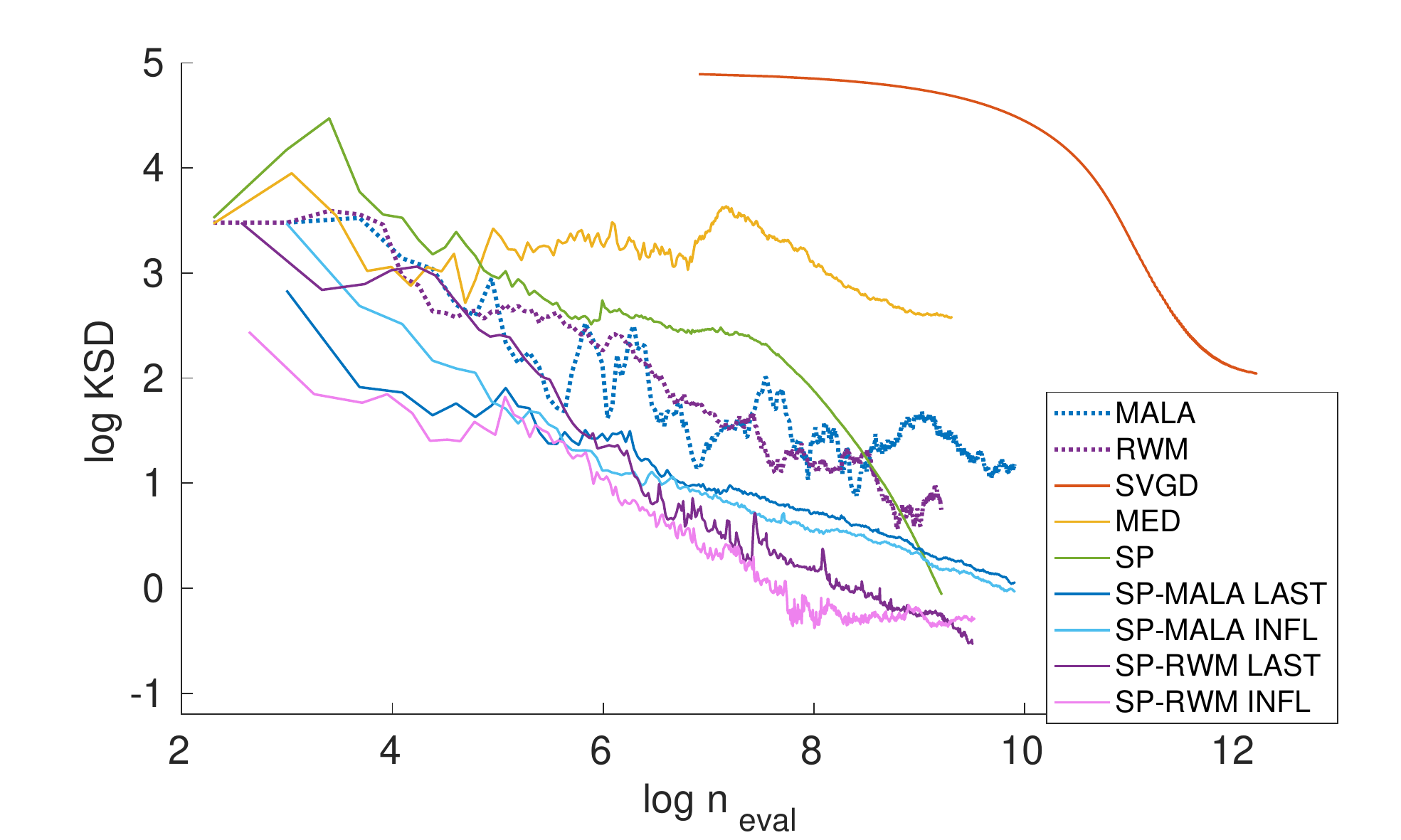}
\caption{$d = 4$}
\label{fig:ODE low dim}
\end{subfigure}
\begin{subfigure}[t]{0.49\textwidth}
\centering
\includegraphics[width = 1.0\textwidth]{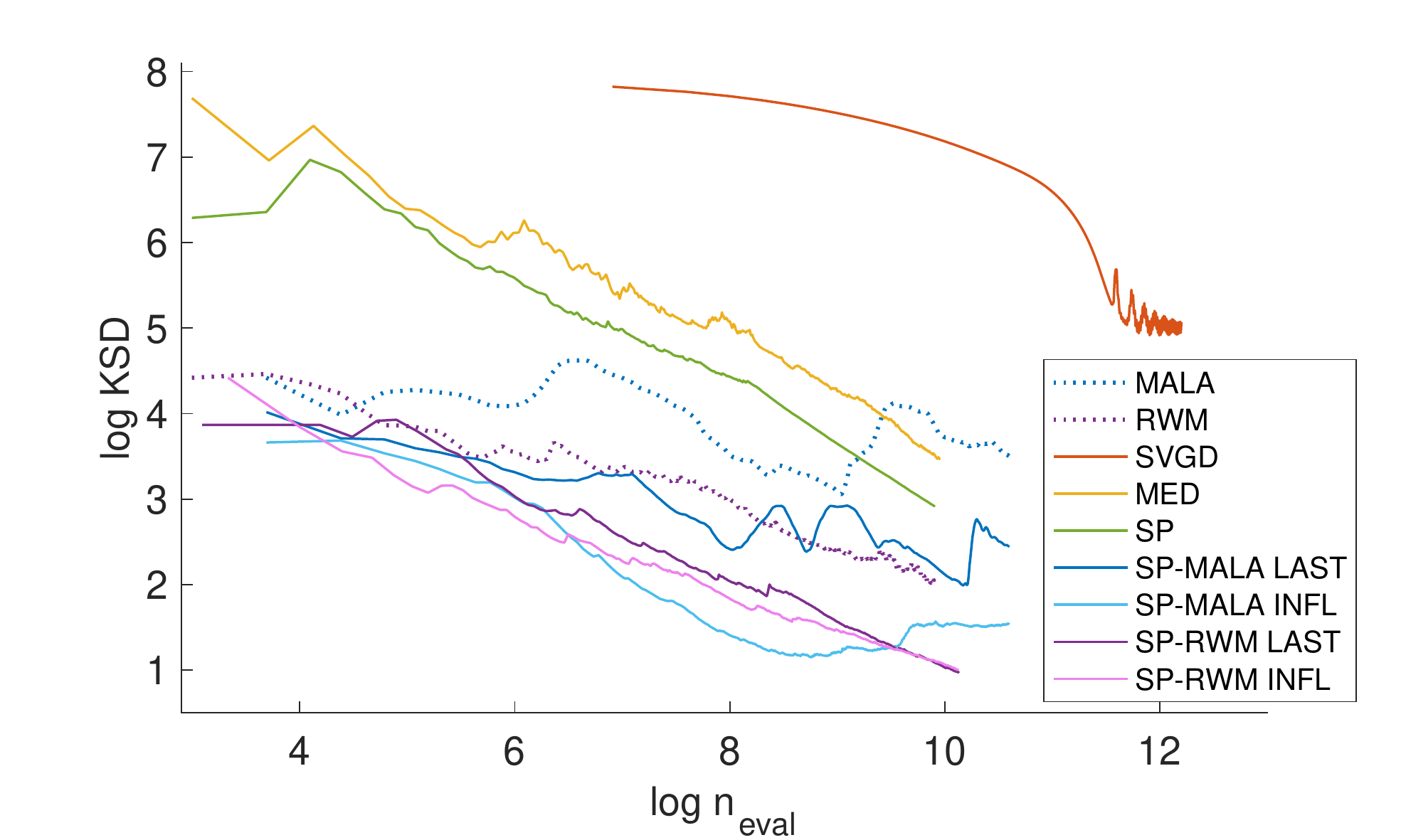}
\caption{$d = 10$}
\label{fig:ODE high dim}
\end{subfigure}
\cprotect\caption{ODE experiment, $d$-dimensional.
The new SP-MCMC method was compared against the original SP method of \cite{Chen2018SteinPoints}, as well as against standard MCMC, MED \citep{RoshanJoseph2015} and SVGD \citep{Liu2016SVGD}.
Each method produced an empirical measure $\frac{1}{n}\sum_{i=1}^n \delta_{x_i}$ whose distance to the target $P$ was quantified by the kernel Stein discrepancy (KSD).
The computational cost was quantified by the number $n_{\text{eval}}$ of times either $\tilde{p}$ or its gradient were evaluated.
}
\end{figure*}

For objectivity, the \emph{energy distance} $E_P$ \cite{Szekely2004,Baringhaus2004} was used to assess closeness of all empirical measures to the target.\footnote{The energy distance $E_P$ is equivalent to MMD based on the conditionally positive definite kernel $k(x,y) = - \|x - y\|_2$ \citep{Sejdinovic2013}.
It was computed using a high-quality empirical approximation of $P$ obtained from a large MCMC output.
}
SP-MCMC was implemented with $m_j = 5 \; \forall j$.
In addition to SP-MCMC, the methods SP, MED, SVGD and standard MCMC were also considered, with implementation described in Appendix \ref{ap: extra results}.
All methods produced a point set of size $n = 1000$.
The results, presented in Fig. \ref{fig:IGARCH}, are indexed by the computational cost of running each method, which is a count of the total number $n_{\text{eval}}$ of times either $\tilde{p}$ or $\nabla \log \tilde{p}$ were evaluated.
It can be seen that SP-MCMC offers improved performance over the original SP method for fixed computational cost, and in turn over both MED and SVGD in this experiment.
Typical point sets produced by each method are displayed in Fig. \ref{fig: IGARCH points}.
The performance of the pre-conditioned kernel on this task was investigated in Appendix \ref{subsubsec: preconditioner test}.

\subsection{System of Coupled ODEs} \label{subsec: ODE}

Our final example is more challenging and offers an opportunity to explore the limitations of SP-MCMC in higher dimensions.
The context is an indirectly observed ODE
\begin{eqnarray*}
y_i & = & g(u(t_i)) + \epsilon_i, \hspace{18pt} \epsilon_i \stackrel{\text{i.i.d.}}{\sim} \mathcal{N}(0,\sigma^2 \mathrm{I}) \\
\dot{u}(t) & = & f_\theta(t,u), \hspace{40pt} u(0) = u_0
\end{eqnarray*}
and, in particular, Bayesian inference for the parameter $\theta$ in the gradient field.
Here $y_i \in \mathbb{R}^p$, $u(t) \in \mathbb{R}^q$ and $\theta \in \mathbb{R}^d$ for $p,q,d \in \mathbb{N}$.
For our experiment, $f_\theta$ and $g$ comprised two instantiations of the Goodwin oscillator \citep{Goodwin1965a}, one low-dimensional with $(q,d) = (2,4)$ and one higher-dimensional with $(q,d) = (8,10)$.
In both cases $p = 2$, $\sigma = 0.1$ and 40 measurements were observed at uniformly-spaced time points in $[41,80]$.
The Goodwin oscillator does not permit a closed form solution, meaning that each evaluation of the likelihood function requires the numerical integration of the ODE at a non-negligible computational cost.
SP-MCMC was implemented with the \verb!INFL! criterion and $m_j = 10$ ($d = 4$), $m_j = 20$ ($d = 10$).
Full details of the ODE and settings for MED and SVGD are provided in Appendix \ref{subsubsec: Goodwin details}.

In this experiment, KSD was used to assess closeness of all empirical measures to the target.\footnote{The more challenging nature of this experiment meant accurate computation of the energy distance was precluded, due to the fact that a sufficiently high-quality empirical approximation of $P$ could not be obtained.}
Naturally, SP and SP-MCMC are favoured by this choice of assessment criterion, as these methods are designed to directly minimise KSD. 
Therefore our main focus here is on the comparison between SP and SP-MCMC.
All methods produced a point set of size $n = 1000$.
Results are shown in Fig. \ref{fig:ODE low dim} (low-dimensional) and Fig. \ref{fig:ODE high dim} (high-dimensional).
Note how the gain in performance of SP-MCMC over SP is more substantial when $d = 10$ compared to when $d = 4$, supporting our earlier intuition for the advantage of local optimisation using a Markov kernel.


\section{Theoretical Results}
\label{sec:theory}

Let $\Omega$ be a probability space on which the collection of random variables $Y_{j,l} : \Omega \rightarrow \mathcal{X}$ representing the $l^{\text{th}}$ state of the Markov chain run at the $j^{\text{th}}$ iteration of SP-MCMC are defined. 
Each of the three algorithms that we consider correspond to a different initialisation of these Markov chains and we use $\mathbb{E}$ to denote expectation over randomness in the $Y_{j,l}$. 
For example, the algorithm called \verb+LAST+ would set $Y_{j,1}(\omega) = x_{j-1}$. 
It is emphasised that the results of this section hold for \emph{any} choice of function \verb!crit! that takes values in $\mathcal{X}$.
As a stepping-stone toward our main result, we first extend the theoretical analysis of the original SP method to the case where the global search is replaced by a Monte Carlo search based on $m_i$ independent draws from $P$ at iteration $i$ of the SP method.
\begin{theorem}[i.i.d.\ SP-MCMC Convergence]
\label{thm:stein-consistent-main-text}
Suppose that the kernel $k_0$ satisfies $\int_{\mathcal{X}} k_0(x,\cdot) \mathrm{d}P(x) \equiv 0$ and $\mathbb{E}_{Z \sim P}[e^{\gamma k_0(Z,Z) }] < \infty$ for some $\gamma > 0$.
Let $(m_j)_{j = 1}^n \subset \mathbb{N}$ be a fixed sequence, and consider idealised Markov chains with $Y_{j,l} \stackrel{\text{\emph{i.i.d.}}}{\sim} P$ for all $1 \leq l \leq m_j$, $j \in \mathbb{N}$.
Let $\{x_i\}_{i=1}^n$ denote the output of SP-MCMC. 
Then, writing $a \wedge b = \min\{a,b\}$, $\exists \; C>0$ such that
\begin{align*}
\mathbb{E}\left[D_{\mathcal{K}_0, P}(\{x_i\}_{i=1}^n)^2\right] \leq \textstyle \frac{C}{n} \sum_{i=1}^n \frac{\log(n \wedge m_i) \wedge \sup_{x \in \mathcal{X}} k_0(x,x)}{n \wedge m_i}.  
\end{align*} 
\end{theorem}
The constant $C$ depends on $k_0$ and $P$, and the proof in Appendix \ref{ap: prf sec 1} makes this dependence explicit.

It follows that SP-MCMC with independent sampling from $P$ is consistent whenever each $m_j$ grows with $n$.
When $m_j = m$ for all $j$ we obtain: 
\begin{eqnarray*}
\mathbb{E}\left[D_{\mathcal{K}_0, P}(\{x_i\}_{i=1}^n)^2\right] \leq \textstyle C \frac{\log(n \wedge m) \wedge \sup_{x \in \mathcal{X}} k_0(x,x)}{n \wedge m} ,
\end{eqnarray*}
and by choosing $m = n$, we recover the rate \eqref{eq: SP theory result} of the original SP algorithm which optimizes over all of $\mathcal{X}$ \cite{Chen2018SteinPoints}.
For bounded kernels, the result improves over the $O(1/n + 1/\sqrt{m})$ independent sampling kernel herding rate established in \citep[App.~B]{Lacoste-julien2015}.
Thm.~\ref{thm:stein-consistent-main-text} more generally accommodates unbounded kernels at the cost of a $\log(n \wedge m)$ factor.

The role of Thm. \ref{thm:stein-consistent-main-text} is limited to providing a stepping stone to Thm. \ref{thm: main theorem}, as it is not practical to obtain exact samples from $P$ in general.
To state our result in the general case, restrict attention to $\mathcal{X} = \mathbb{R}^d$, consider a function $V: \mathcal{X} \to [1,\infty)$ and define the associated operators
$\| f \|_V := \sup_{x \in \mathcal{X}}| f(x) |/V(x)$, $\| \mu \|_{V} := \sup_{f: \| f\|_V \leq 1} | \int f \mathrm{d}\mu |$ respectively on functions $f : \mathcal{X} \rightarrow \mathbb{R}$ and on signed measures $\mu$ on $\mathcal{X}$.
A Markov chain $(Y_i)_{i \in \mathbb{N}} \subset \mathcal{X}$ with $n$th step transition kernel $\mathrm{P}^n$ is called \emph{$V$-uniformly ergodic} \citep[][Chap. 16]{Meyn2012} if $\exists R \in [0,\infty), \rho \in (0,1)$ such that $\| \mathrm{P}^n(y,\cdot)-P \|_{V} \leq R V(y) \rho^n$ for all initial states $y \in \mathcal{X}$ and all $n \in \mathbb{N}$.
The proof of the following is provided in Appendix \ref{ap: prf sec 2}:

\begin{theorem}[SP-MCMC Convergence]
\label{thm: main theorem}
Suppose $\int_{\mathcal{X}} k_0(x,\cdot) \mathrm{d}P(x) \equiv 0$ with $\E_{Z \sim P}[e^{\gamma k_0(Z, Z)}] < \infty$ for $\gamma > 0$.
For a sequence $(m_j)_{j = 1}^n \subset \mathbb{N}$, let $\{x_i\}_{i=1}^n$ denote the output of SP-MCMC, based on time-homogeneous reversible Markov chains $(Y_{j,l})_{l = 1}^{m_j}$, $j \in \mathbb{N}$, generated using the same $V$-uniformly ergodic transition kernel.
Define 
$
V_\pm(s) := \sup_{x : k_0(x,x) \leq s^2} k_0(x,x)^{1/2}V(x)^{\pm 1}
$
and $S_i = \sqrt{2\log(n \wedge m_i)/\gamma}$.
Then $\exists \; C>0$ such that
\begin{align*}
\mathbb{E}\left[D_{\mathcal{K}_0, P}(\{x_i\}_{i=1}^n)^2\right] \leq \textstyle \frac{C}{n} \sum_{i=1}^n \frac{S_i^2}{n} + \frac{ V_+(S_i) V_-(S_i)}{m_i} .
\end{align*}
\end{theorem}

We give an example of verifying the preconditions of Thm. \ref{thm: main theorem} for MALA.
Let $\pset$ denote the set of distantly dissipative\footnote{The target $P$ is said to be \emph{distantly dissipative} \citep{Eberle2016,Gorham2016} if $\kappa_0 \stackrel{\Delta}{=} \lim \inf_{r \rightarrow \infty} \kappa(r) > 0$ for
$
\kappa(r) = \textstyle \inf \left\{ -2 \frac{\langle \nabla \log [\tilde{p}(x) - \tilde{p}(y)] , x - y \rangle}{\|x-y\|_2^2} : \|x-y\|_2 = r \right\}.
$} distributions with $\grad \log p$ Lipschitz on $\mathcal{X} = \mathbb{R}^d$.
Let $C_b^{(r,r)}$ be the set of functions $k : \mathbb{R}^d \times \mathbb{R}^d \rightarrow \mathbb{R}$ with $(x,y) \mapsto \nabla_x^l \nabla_y^l k(x,y)$ continuous and uniformly bounded for $l \in \{0,\dots,r\}$.
Let $q(x,y)$ be a density for the proposal distribution of MALA, and let $\alpha(x,y)$ denote the acceptance probability for moving from $x$ to $y$, given that $y$ has been proposed.
Let $A(x) = \{y \in \mathcal{X} : \alpha(x,y) = 1 \}$ denote the region where proposals are always accepted and let $R(x) = \mathcal{X} \setminus A(x)$.
Let $I(x) := \{y : \|y\|_2 \leq \|x\|_2 \}$.
MALA is said to be \emph{inwardly convergent} \citep[][Sec.~4]{Roberts1996} if
\begin{align}
\lim_{\|x\|_2 \rightarrow \infty} \int_{A(x) \Delta I(x)} q(x,y) \mathrm{d}y = 0 \label{asm: inw conv}
\end{align}
where $A\Delta B$ denotes the symmetric set difference $(A \cup B) \setminus (A \cap B)$. 
The proof of the following is provided in Appendix \ref{ap: prf mala consistency}:

\begin{theorem}[SP-MALA Convergence]
\label{thm: SPMCMC consistent}
Suppose $k_0$ has the form \cref{eq:stein_kernel}, based on a kernel $k \in C_b^{(1,1)}$ and a target $P\in\pset$ such that $\int_{\mathcal{X}} k_0(x,\cdot) \mathrm{d}P(x) \equiv 0$. 
Let $(m_j)_{j = 1}^n \subset \mathbb{N}$ be a fixed sequence and let $\{x_i\}_{i=1}^n$ denote the output of SP-MCMC, based on Markov chains $(Y_{j,l})_{l = 1}^{m_j}$, $j \in \mathbb{N}$, generated using MALA transition kernel with step size $h$ sufficiently small.
Assume $P$ is such that MALA is inwardly convergent.
Then MALA is $V$-uniformly ergodic for $V(x) = 1 + \|x\|_2$ and $\exists \; C>0$ such that
\begin{eqnarray*}
\mathbb{E}\left[D_{\mathcal{K}_0, P}(\{x_i\}_{i=1}^n)^2\right] \leq \textstyle \frac{C}{n} \sum_{i=1}^n \frac{\log(n \wedge m_i)}{n \wedge m_i} .
\end{eqnarray*}
\end{theorem}

Our final result, proved in Appendix \ref{ap: prf sec 3}, establishes that the pre-conditioner kernel proposed in Sec. \ref{subsec: precondition} can control weak congergence to $P$ when the pre-conditionner $\Lambda$ is symmetric positive definite (denoted $\Lambda \psdgt 0$).
It is a generalisation of Thm.~8 of \citet{Gorham2017}, who treated the special case of $\Lambda = I$:

\begin{theorem}[Pre-conditioned IMQ KSD Controls Convergence]
\label{thm:pre-conditioning}
Suppose $k_0$ is a Stein kernel \cref{eq:stein_kernel} for a target $P\in\pset$ and a pre-conditioned IMQ base kernel \cref{eq: IMQ kernel} with $\beta \in (-1,0)$ and $\Lambda \psdgt 0$. 
If $D_{\mathcal{K}_0,P}(\{x_i\}_{i=1}^n) \rightarrow 0$ then $\frac{1}{n} \sum_{i=1}^n \delta_{x_i}$ converges weakly to $P$.
\end{theorem}


\section{Conclusion}
\label{sec:conclusion}
This paper proposed fundamental improvements to the SP method of \cite{Chen2018SteinPoints}, establishing, in particular, that the global search used to select each point can be replaced with a finite-length sample path from an MCMC method.
The convergence of the proposed SP-MCMC method was established, with an explicit bound provided on the KSD in terms of the $V$-uniform ergodicity of the Markov transition kernel.

Potential extensions to our SP-MCMC method include the use of fast approximate Markov kernels for $P$ (such as the unadjusted Langevin algorithm; see Appendix \ref{ap: prf mala consistency}), fast approximations to KSD \citep{Jitkrittum2017,Huggins2018}, exploitation of conditional independence structure in $P$ \citep{Wang2018,Zhuo2018} and extension to a general Riemannian manifold $\mathcal{X}$ \citep{Liu2018a,Barp2018RS}. 
One could also attempt to use our MCMC optimization approach to accelerate related algorithms such as kernel herding \cite{Chen2010,Bach2012,Lacoste-julien2015}. 
Other recent approaches to quantisation in the Bayesian context include \cite{Futami2018,Hu2018,Frogner2018,Zhang2018,Chen2018,Li2019}, and an assessment of the relative performance of these methods would be of interest.
However, we note that these approaches are not accompanied by the same level of theoretical guarantees that we have established.


\section*{Acknowledgements}

The authors are grateful to the reviewers for their critical feedback on the manuscript.
WYC was supported by the Australian Research Council Centre of Excellence for Mathematics and Statistical Frontiers.
AB was supported by a Roth Scholarship from the Department of Mathematics at Imperial College London, UK.
WYC, AB, FXB, MG and CJO were supported by the Lloyd's Register Foundation programme on data-centric engineering at the Alan Turing Institute, UK.
MG was supported by the EPSRC grants [EEP/P020720/1, EP/R018413/1, EP/R034710/1, EP/R004889/1] and a Royal Academy of Engineering Research Chair.




\newpage 
\onecolumn
\appendix

\setcounter{table}{0}  
\renewcommand{\thetable}{S\arabic{table}} 
\setcounter{figure}{0} 
\renewcommand{\thefigure}{S\arabic{figure}}

\section{Supplementary Material}

In this supplement, Sections \ref{ap: prf sec 1}, \ref{ap: prf sec 2}, \ref{ap: prf mala consistency} and \ref{ap: prf sec 3} contain complete proofs for the results stated in the main text, Section \ref{ap: mala explain} details the Markov transition kernel that was used in all experiments and Section \ref{ap: extra results} contains an in-depth presentation of our empirical results which were summarised at a high level in the main text.

\subsection{Proof of Theorem \ref{thm:stein-consistent-main-text}} \label{ap: prf sec 1}

Note that the performance of greedy algorithms for minimisation of MMD was studied in \cite{DeMarchi2005,Santin2017}.
Our analysis, and that in \cite{Chen2018SteinPoints}, differ in several respects from this work - not least in that our arguments do not require the set $\mathcal{X}$ to be compact.

First, we state and prove a generalisation of Theorem 5 in \cite{Chen2018SteinPoints}, which quantifies KSD convergence for point sets produced by approximately optimizing over arbitrary subsets of $\mathcal{X}$:

\begin{theorem}[Generalized Stein Point Convergence]
\label{thm:stein-point-convergence}
Suppose $k_0$ is a reproducing kernel with $\int_{\mathcal{X}} k_0(x,\cdot) \mathrm{d}P(x) \equiv 0$.
Fix $n \in \mathbb{N}$, and, for each $1 \leq j \leq n$, any $\mathcal{Y}_j \subseteq \mathcal X$ and any $h_j$ in the convex hull of $\{ k_0(x,\cdot) \}_{x\in \mathcal{Y}_j}$.
Fix $\delta > 0$ and, for each $1 \leq i \leq n$, fix $S_i \geq 0$ and $r_i > 0$.
Any point set $\{x_i\}_{i=1}^n \subset \mathcal{X}$ satisfying 
\begin{align}
\frac{k_0(x_j,x_j)}{2}&+\displaystyle\sum_{i=1}^{j-1} k_0(x_i, x_j) 
	\leq \frac{\delta}{2} + \frac{S_j^2}{2}
	+ \displaystyle\inf_{x\in \mathcal{Y}_j}\displaystyle\sum_{i=1}^{j-1} k_0(x_i, x)  \label{eq: close enough}
\end{align}
for each $1 \leq j \leq n$ also satisfies 
\begin{align}
D_{\mathcal{K}_0,P}(\{x_i\}_{i=1}^n) 
	\leq \exp \left( \frac 1 2 \sum_{j=1}^n \frac 1 {r_j} \right) \sqrt{\frac{\delta}{n}+\frac{1}{n^2}\sum_{i=1}^{n}\left(S_i^2 + r_i\|h_i\|_{\mathcal{K}_0}^2\right)}. \label{eq: thm4 result}
\end{align}
\end{theorem}
\begin{proof}
Let  $a_n := n^2 D_{\mathcal{K}_0,P}(\{x_i\}_{i=1}^n)^2 = \sum_{i=1}^n \sum_{i'=1}^n k_0(x_i,x_{i'}) = \| \sum_{i=1}^n k_0(x_i, \cdot) \|^2_{\mathcal K_0}$. Then 
\begin{align}
a_n \; =\; \sum_{i=1}^n\sum_{i'=1}^n k_0(x_i,x_{i'}) \;=\;
 a_{n-1}+ k_0(x_n,x_n) +2\sum_{i=1}^{n-1}k_0(x_i,x_n) \; \leq \; a_{n-1}+ \delta  +  S_n^2+ 2 \inf_{x\in \mathcal{Y}_n} \sum_{i=1}^{n-1}k_0(x_i,x),  \label{eq: an recursion}
\end{align}
where for the final inequality we have used \eqref{eq: close enough} with $j = n$. 
Next, we let $f_n := \sum_{i=1}^n k_0(x_i, \cdot)$ so that $\|f_n\|_{\mathcal{K}_0} = \sqrt{a_n}$. 
Applying in the first instance the Cauchy-Schwarz inequality, then making use of the arithmetic-geometric inequality\footnote{Recall that the arithmetic-geometric mean inequality states that for any constants $b_1,\ldots,b_m \geq 0$, $\frac{1}{m}\sum_{k=1}^m b_k \geq (\prod_{k=1}^m b_k)^{\frac{1}{m}}$.}, we get:
\begin{align}
2 \inf_{x \in \mathcal{Y}_n} f_{n-1}(x) \;= \;
  2 \inf_{f \in \mathcal M_n} \langle f_{n-1},f \rangle_{\mathcal K_0} \; & \leq \; 2  \langle f_{n-1},h_n \rangle_{\mathcal K_0} \nonumber \\
 & \leq 2 \; \sqrt{\|f_{n-1}\|^2_{\mathcal K_0} \|h_n\|^2_{\mathcal K_0} } \; = 2 \; \sqrt{\left(\frac{\|f_{n-1}\|^2_{\mathcal K_0}}{r_n}\right) (r_n \|h_n\|^2_{\mathcal K_0} )} \nonumber \\
 & \leq \; r_n \| h_n \|_{\mathcal K_0}^2 + \frac{a_{n-1}}{r_n}.   \label{eq: a-g ineq}
\end{align}
Combining \eqref{eq: an recursion} and \eqref{eq: a-g ineq} establishes the recurrence relation
\begin{align}
a_n \leq \left(1 + \frac{1}{r_n} \right) a_{n-1}+\delta +S_n^2 +r_n \| h_n \|_{\mathcal K_0}^2.  \label{eq: recurrance relation}
\end{align}
Expanding the recurrence leads to a product of terms of the form $(1 + \frac{1}{r_n})$ which must be controlled.
To this end, we use the fact that $\log(1+x) \leq x $ for $x\geq0$ implies that
$$  
\log \prod_{j=1}^n\left(1 + \frac{1}{r_{j}} \right)  
\; = \;
\sum_{j=1}^n \log \left(1 + \frac{1}{r_{j}} \right)
 \; \leq  \; 
\sum_{j=1}^n \frac{1}{r_{j}}, 
$$ 
and, noting that the function $i \mapsto \prod_{j=1}^i (1 + 1/r_{n-j+1})$ is increasing, we can bound the product
$$
\prod_{j=1}^i \left(1 + \frac{1}{r_{n-j+1}} \right)
\; \leq \;
\prod_{j=1}^n \left(1 + \frac{1}{r_{j}} \right) 
 \; \leq \;
  \exp \left( \sum_{j=1}^n \frac 1 {r_j} \right),
 $$
uniformly in $i$.
This implies that the recurrence relation in \eqref{eq: recurrance relation} satisfies
\begin{eqnarray*} 
a_n 
& \leq & 
\sum_{i=0}^{n-1} \left( \delta + S_{n-i}^2 + r_{n-i} \| h_{n-i} \|_{\mathcal K_0}^2 \right) \prod_{j=1}^i\left(1 + \frac{1}{r_{n-j+1}} \right) \\ 
& \leq & 
 \exp \left( \sum_{j=1}^n \frac 1 {r_j} \right)\sum_{i=0}^{n-1} \left( \delta + S_{n-i}^2 + r_{n-i} \| h_{n-i} \|_{\mathcal K_0}^2 \right), \end{eqnarray*}
from which the result is established.
\end{proof}

Theorem \ref{thm:stein-point-convergence} is a refinement of the argument used in the first part of the proof of Theorem 5 in \cite{Chen2018SteinPoints}.
It serves to make explicit the roles of $S_j$ and $\|h_j\|_{\mathcal{K}_0}$ and distinguishes between the content of \eqref{eq: thm4 result} and subsequent assumptions on $k_0$ and $P$ that are used to bound the terms that are involved.

The result of Theorem \ref{thm:stein-point-convergence} provides an upper bound in the situation where the sets $\mathcal{Y}_j$ are fixed.
To make use of Theorem \ref{thm:stein-point-convergence} in the context of SP-MCMC, where the sets $\mathcal{Y}_j$ are instead randomly generated, we must therefore establish probabilistic bounds on the quantities $S_i$ and $\|h_i\|_{\mathcal{K}_0}$ that appear in the statement of Theorem \ref{thm:stein-point-convergence}.
This is the content of the next result:

\begin{theorem}[Generalized i.i.d.\ SP-MCMC Convergence] \label{thm: SP ideal MC general result}
Suppose $k_0$ is a reproducing kernel with $\int_{\mathcal{X}} k_0(x,\cdot) \mathrm{d}P(x) \equiv 0$ and $\mathbb{E}_{Z \sim P}[e^{\gamma k_0(Z,Z) }] < \infty$.
Fix a sequence $(m_j)_{j=1}^n \subset \mathbb{N}$ and, for each $j \in \mathbb{N}$, let $\mathcal{Y}_j$ be the set of independent random variables $\{Y_{j,l}\}_{l=1}^{m_j}$, with each $Y_{j,l} \sim P$. 
For each $1 \leq i \leq n$, fix $\tilde{S}_i \geq 0$ and $r_i > 0$.
Then $\exists C > 0$ such that
\begin{eqnarray}
\E\left[D_{\mathcal{K}_0, P}(\{x_i\}_{i=1}^n)^2\right]  
\; \leq \; C \exp\Bigg( \sum_{j=1}^n \frac{1}{r_j} \Bigg) \frac{1}{n^2} \sum_{i=1}^n \left( r_i e^{- \frac{\gamma}{2} \tilde{S}_i^2} + \left( 1 + \frac{r_i}{m_i} \right) \min\left\{\tilde{S}_i^2 , \sup_{x \in \mathcal{X}} k_0(x,x) \right\}  \right) ,  \label{eq: gen iid spmcmc result}
\end{eqnarray}
where in each case the total expectation $\mathbb{E}$ is taken over realisations of the random sets $\mathcal{Y}_j$, $j \in \mathbb{N}$.
\end{theorem}
\begin{proof}

\def\indic#1{\mathbb{I}\left[{#1}\right]} 
\def\norm#1{\|{#1}\|} 

Recall that the $j$th iteration of SP-MCMC requires that random variables $(Y_{j,l})_{l=1}^{m_j}$ are instantiated.
Define the set $\mathcal{Y}_j$ to consist of the subset of these $m_j$ samples for which $Y_{j,l} \in B_j := \{ x \in \mathcal{X} : k_0(x,x) \leq \tilde{S}_j^2\}$ is satisfied. 
Note that SP-MCMC selects the $j$th point $x_j$ from the collection $\{Y_{j,l}\}_{l=1}^{m_j}$ such that
\begin{align*}
\frac{k_0(x_j,x_j)}{2}+\displaystyle\sum_{i=1}^{j-1} k_0(x_i, x_j) 
	&= \displaystyle\inf_{x \in \{ Y_{j,l} : l = 1,\dots,m_j \}}\frac{k_0(x,x)}{2}
	+\displaystyle\sum_{i=1}^{j-1} k_0(x_i, x)   \\
	&\leq \displaystyle\inf_{x\in \mathcal{Y}_j}\frac{k_0(x,x)}{2}
	+\displaystyle\sum_{i=1}^{j-1} k_0(x_i, x)
	\; \leq \; \frac{1}{2}\min\left\{\tilde{S}_j^2 , \sup_{x \in \mathcal{X}} k_0(x,x) \right\}+ \displaystyle\inf_{x\in \mathcal{Y}_j} 	\displaystyle\sum_{i=1}^{j-1} k_0(x_i, x). 
\end{align*}
so that \eqref{eq: close enough} is satisfied with $\delta = 0$ and $S_j := \min\left\{ \tilde{S}_j , \sup_{x \in \mathcal{X}} k_0(x,x)^{1/2} \right\}$.

Let $h_j(\cdot) := \frac{1}{m_j} \sum_{l=1}^{m_j} k_0(Y_{j,l}, \cdot ) \mathbb{I}[Y_{j,l} \in B_j]$, which is an element of the convex hull of $\{k_0(x,\cdot)\}_{x \in \mathcal{Y}_j}$.
Define also the truncated kernel mean embeddings
\begin{align*}
k_j^-(\cdot) := \int k_0(x,\cdot) \mathbb{I}[x \in B_j] \mathrm{d}P(x), \qquad k_j^+(\cdot) := \int k_0(x,\cdot) \mathbb{I}[x \notin B_j] \mathrm{d}P(x).
\end{align*}
From the triangle inequality followed by Jensen's inequality
\begin{align}
\|h_j\|_{\mathcal{K}_0}^2 &\leq 2 \left( \| k_j^- \|_{\mathcal{K}_0}^2 + \|h_j - k_j^- \|_{\mathcal{K}_0}^2 \right) .  \label{eq: triangle ineq h trun IID}
\end{align}
In what follows we aim to bound the two terms on the right hand side of \eqref{eq: triangle ineq h trun IID}.

\paragraph{Bound on $\|k_j^{-}\|_{\mathcal{K}_0}^2$:}

For the first term in \eqref{eq: triangle ineq h trun IID}, since $\int_{\mathcal{X}} k_0(x,\cdot) \mathrm{d}P(x) \equiv 0$ we have $k_j^+ = - k_j^-$.
Thus, we deduce that
\begin{align*}
\|k_j^-\|_{\mathcal{K}_0}^2 = \|k_j^+\|_{\mathcal{K}_0}^2
	&= \iint k_0(x,y) \mathbb{I}[x \notin B_j] \mathrm{d}P(x) \mathbb{I}[y \notin B_j] \mathrm{d}P(y) \\
	&\leq \left(\int \sqrt{k_0(x,x)}\mathbb{I}[x \notin B_j] \mathrm{d}P(x) \right)^2 
	\; \leq \; \int k_0(x,x)\mathbb{I}[x \notin B_j] \mathrm{d}P(x)
\end{align*}
where the final two inequalities follow by Cauchy-Schwarz and Jensen's inequality.
Now, let $Y = k_0(Z,Z)$ for $Z \sim P$ and $b := \mathbb{E}[e^{\gamma Y}] < \infty$.
Following Appendix A.1.3 of \cite{Chen2018SteinPoints}, we will bound the tail expectation above by considering the biased random variable $Y^* = k_0(Z^*,Z^*)$
for $Z^*$ with density 
$$
\rho(z^*) = \frac{k_0(z^*,z^*) p(z^*)}{\mathbb{E}[Y]}.
$$
To this end we have, by the relation $x \leq e^x$,
\[
\mathbb{E}[e^{\frac{\gamma}{2} Y^*}] 
	= \mathbb{E}[e^{\frac{\gamma}{2} k_0(Z^*,Z^*)}] 
	= \frac{\mathbb{E}[k_0(Z,Z)e^{\frac{\gamma}{2} k_0(Z,Z)}]}{\mathbb{E}[Y]}
	= \frac{\mathbb{E}[\frac{\gamma}{2} Y e^{\frac{\gamma}{2} Y}] }{\frac{\gamma}{2} \mathbb{E}[Y]}
	\leq \frac{ \mathbb{E}[e^{\gamma Y}] }{\lambda \mathbb{E}[Y]}
	= \frac{2b}{\gamma \mathbb{E}[Y]} .
\]
From an application of Markov's inequality we see that
$$
\mathbb{P}[Y^* \geq \tilde{S}_j^2] = \mathbb{P}[e^{\frac{\gamma}{2} Y^*} \geq e^{\frac{\gamma}{2} \tilde{S}_j^2}] \leq \frac{\mathbb{E}[e^{\frac{\gamma}{2} Y^*}]}{e^{\frac{\gamma}{2} \tilde{S}_j^2}} \leq \frac{2b}{\gamma \mathbb{E}[Y]} e^{-\frac{\gamma}{2} \tilde{S}_j^2}
$$
and as a consequence
\begin{align}
\|k_j^-\|_{\mathcal{K}_0}^2
	\leq \int k_0(x,x)\mathbb{I}[k_0(x,x) > \tilde{S}_j^2] \mathrm{d}P(x)
	= {\mathbb{E}[Y]} \int \mathbb{I}[k_0(x,x) > \tilde{S}_j^2] \rho(x)\mathrm{d}x
	= {\mathbb{E}[Y]} \mathbb{P}[Y^* \geq \tilde{S}_j^2]
	\leq \frac{2b}{\gamma} e^{-\frac{\gamma}{2} \tilde{S}_j^2} \label{eq: kjminus bound} .
\end{align}

\paragraph{Bound on $\|h_j - k_j^{-}\|_{\mathcal{K}_0}^2$:}

For the second term in \eqref{eq: triangle ineq h trun IID}, we have that
\begin{eqnarray*}
\|h_j - k_j^-\|_{\mathcal{K}_0}^2 & = & \frac{1}{m_j^2} \sum_{l,l' = 1}^{m_j} k_0(Y_{j,l},Y_{j,l'}) \mathbb{I}[Y_{j,l}, Y_{j,l'} \in B_j] - \frac{2}{m_j} \sum_{l=1}^{m_j} \int k_0(x, Y_{j,l}) \mathbb{I}[x, Y_{j,l} \in B_j] \mathrm{d}P(x) \\
& & + \iint k_0(x,x') \mathbb{I}[x,x' \in B_j] \mathrm{d}P(x) \mathrm{d}P(x') \\
& = & \frac{1}{m_j} \sum_{l=1}^{m_j} \left\{ \frac{1}{m_j} \sum_{l'=1}^{m_j} k_0(Y_{j,l},Y_{j,l'}) \mathbb{I}[Y_{j,l}, Y_{j,l'} \in B_j] - \int k_0(x, Y_{j,l}) \mathbb{I}[x, Y_{j,l} \in B_j] \mathrm{d}P(x) \right\} \\
& & - \int \left\{ \frac{1}{m_j} \sum_{l = 1}^{m_j} k_0(x,Y_{j,l}) \mathbb{I}[x,Y_{j,l} \in B_j] - \int k_0(x,x') \mathbb{I}[x,x' \in B_j] \mathrm{d}P(x') \right\} \mathrm{d}P(x) \\
& = & \frac{1}{m_j^2} \sum_{l=1}^{m_j} \sum_{l'=1}^{m_j} \left\{ h_{Y_{j,l'}}(Y_{j,l}) - \int h_{Y_{j,l}}(x) \mathrm{d}P(x) \right\} - \frac{1}{m_j} \sum_{l = 1}^{m_j} \int \left\{ h_x(Y_{j,l}) - \int h_x(x') \mathrm{d}P(x') \right\} \mathrm{d}P(x) \\
& = & \frac{1}{m_j^2} \sum_{l=1}^{m_j} \sum_{l'=1}^{m_j} \left\{ h_{Y_{j,l'}}(Y_{j,l}) - \int h_{Y_{j,l'}}(x) \mathrm{d}P(x) \right\} - \frac{1}{m_j} \sum_{l = 1}^{m_j} \int \left\{ h_x(Y_{j,l}) - \int h_x(x') \mathrm{d}P(x') \right\} \mathrm{d}P(x)
\end{eqnarray*}
where $h_x(x') := k_0(x,x') \mathbb{I}[x,x' \in B_j]$.
Thus, letting again $Z \sim P$ be independent of all other random variables that we have defined,
\begin{eqnarray}
\mathbb{E} \left[ \|h_j - k_j^-\|_{\mathcal{K}_0}^2 \right] = \frac{1}{m_j^2} \sum_{l=1}^{m_j} \sum_{l'=1}^{m_j} \left\{ \mathbb{E}[ h_{Y_{j,l'}}(Y_{j,l}) ] - \mathbb{E}[ h_{Y_{j,l'}}(Z) ] \right\} - \frac{1}{m_j} \sum_{l = 1}^{m_j} \int \left\{ \mathbb{E}[ h_x(Y_{j,l}) ] - \mathbb{E}[ h_x(Z) ] \right\} \mathrm{d}P(x) . \label{eq: truncated embedding error IID}
\end{eqnarray}
Since the $Y_{j,l}$ are assumed to be independent and distributed according to $P$, all of the terms in \eqref{eq: truncated embedding error IID} vanish apart from the diagonal terms in the first sum, and moreover all of the diagonal terms are identical:
\begin{eqnarray*}
\mathbb{E} \left[ \|h_j - k_j^-\|_{\mathcal{K}_0}^2 \right] = \frac{1}{m_j^2} \sum_{l=1}^{m_j} \left\{ \mathbb{E}[ h_{Y_{j,l}}(Y_{j,l}) ] - \mathbb{E}[ h_{Y_{j,l}}(Z) ] \right\} = \frac{1}{m_j}  \left\{ \mathbb{E}[ h_{Y_{j,1}}(Y_{j,1}) ] - \mathbb{E}[ h_{Y_{j,1}}(Z) ] \right\} .
\end{eqnarray*}
An application of the triangle inequality and Cauchy-Schwarz leads us to the bound
\begin{eqnarray*}
\left| \mathbb{E} \left[ \|h_j - k_j^-\|_{\mathcal{K}_0}^2 \right] \right| & \leq & \frac{1}{m_j}  \left\{ \left| \mathbb{E}[ h_{Y_{j,1}}(Y_{j,1}) ] \right| + \left| \mathbb{E}[ h_{Y_{j,1}}(Z) ] \right| \right\} \\
& = & \frac{1}{m_j} \left\{ \left| \mathbb{E}[k_0(Y_{j,1},Y_{j,1}) \mathbb{I}[Y_{j,1} \in B_j] ] \right| + \left| \mathbb{E}[k_0(Y_{j,1},Z) \mathbb{I}[Y_{j,1},Z \in B_j] ] \right| \right\} \\
& \leq & \frac{1}{m_j} \left\{ \mathbb{E}[k_0(Y_{j,1},Y_{j,1}) \mathbb{I}[Y_{j,1} \in B_j] ]  + \mathbb{E}[k_0(Y_{j,1},Y_{j,1})^{\frac{1}{2}} k_0(Z,Z)^{\frac{1}{2}} \mathbb{I}[Y_{j,1},Z \in B_j] ] \right\} \\
& \leq &  \frac{2}{m_j} \min\left\{\tilde{S}_j^2 , \sup_{x \in \mathcal{X}} k_0(x,x) \right\} . 
\end{eqnarray*}

\paragraph{Overall Bound:}

Combining our bounds for the terms in \eqref{eq: triangle ineq h trun IID} leads to
\begin{eqnarray}
\mathbb{E}[ \|h_j\|_{\mathcal{K}_0}^2 ] \; \leq \; 2 \left( \| k_j^- \|_{\mathcal{K}_0}^2 + \mathbb{E}\left[ \|h_j - k_j^- \|_{\mathcal{K}_0}^2 \right] \right) 
\; \leq \; \frac{4b}{\gamma} e^{-\frac{\gamma}{2} \tilde{S}_j^2} +  \frac{4}{m_j} \min\left\{ \tilde{S}_j^2 , \sup_{x \in \mathcal{X}} k_0(x,x) \right\} .  \label{eq: EH bound IID}
\end{eqnarray}
Finally, we square the conclusion \eqref{eq: thm4 result} of Theorem \ref{thm:stein-point-convergence} (with, recall, $\delta = 0$ and $S_j := \min\{ \tilde{S}_j , \sup_{x \in \mathcal{X}} k_0(x,x)^{1/2} \}$) and take expectations, which combine with \eqref{eq: EH bound IID} to produce \eqref{eq: gen iid spmcmc result} with $C = \max\left\{ \frac{4b}{\gamma} , 4 \right\}$.
\end{proof}

Now we are in a position to prove Theorem \ref{thm:stein-consistent-main-text}, which follows as a specific instance of Theorem \ref{thm: SP ideal MC general result}:

\begin{proof}[Proof of Theorem \ref{thm:stein-consistent-main-text}]

The result follows as a special case of the general result of Theorem \ref{thm: SP ideal MC general result} with $S_i = \sqrt{\frac{2}{\gamma} \log(n \wedge m_i)}$ and $r_i = n$ for $i = 1,\dots,n$.
Indeed, with these settings we can bound the conclusion \eqref{eq: gen iid spmcmc result} of Theorem \ref{thm: SP ideal MC general result} as follows (with $C$ a generic constant):
\begin{eqnarray*}
\mathbb{E}\left[D_{\mathcal{K}_0, P}(\{x_i\}_{i=1}^n)^2\right] 
& \leq & C \exp\Bigg( \sum_{j=1}^n \frac{1}{r_j} \Bigg) \frac{1}{n^2} \sum_{i=1}^n \left( r_i e^{- \frac{\gamma}{2} S_i^2} +  \left( 1 + \frac{r_i}{m_i} \right) \min\left\{S_i^2 , \sup_{x \in \mathcal{X}} k_0(x,x) \right\}  \right) \\
& = & C \frac{1}{n^2} \sum_{i=1}^n \left( \frac{n}{n \wedge m_i} +  \left( 1 + \frac{n}{m_i} \right) \min \left\{ \frac{2}{\gamma} \log(n \wedge m_i) , \sup_{x \in \mathcal{X}} k_0(x,x) \right\}  \right) \\
& \leq & C \frac{1}{n} \sum_{i=1}^n \left( \frac{1}{n \wedge m_i} +  \frac{1}{n \wedge m_i} \min \left\{ \log(n \wedge m_i) , \sup_{x \in \mathcal{X}} k_0(x,x) \right\}  \right) \\
& \leq & C \frac{1}{n} \sum_{i=1}^n  \frac{1}{n \wedge m_i} \min \left\{ \log(n \wedge m_i) , \sup_{x \in \mathcal{X}} k_0(x,x) \right\} 
\end{eqnarray*}
as claimed.
\end{proof}

To conclude this section, we remark that the general result established in Thm. \ref{thm: SP ideal MC general result} also implies conditions under which Monte Carlo search strategies can be successfully applied to the \emph{Stein Herding} algorithm proposed in \cite{Chen2018SteinPoints}.
However, our focus on the greedy version of SP in this work was motivated by the stronger theoretical guarantees posessed by the greedy method, as well as the superior empirical performance reported in \cite{Chen2018SteinPoints}.


\subsection{Proof of Theorem \ref{thm: main theorem}} \label{ap: prf sec 2}

Now we turn to the main task of establishing consistency of SP-MCMC in the Markov chain context.
Necessarily, any quantitative result must depend on mixing properties of the Markov chain being used.
In this research we focused on time-homogeoeus Markov chains and the notion of mixing called \emph{$V$-uniform ergodicity}, defined in Sec. \ref{sec:theory} of the main text.
Recall that, for a function $V: \mathcal{X} \to [1,\infty)$, $V$-uniform ergodicy is the property that
$$ \| \mathrm{P}^n(y,\cdot)-P \|_{V} \leq R V(y) \rho^n$$
for some $R \in [0,\infty)$ and $\rho \in (0,1)$ and for all initial states $y \in \mathcal{X}$ and all $n \in \mathbb{N}$.
The assumption of $V$-uniform ergodicity enables us to provide results that hold for \emph{any} choice of function \verb!crit! that takes values in $\mathcal{X}$.
This includes the functions \verb!LAST!, \verb!RAND! and \verb!INFL! from the main text, but in general the value of \verb!crit!$(\{x_i\}_{i=1}^j)$ is not restricted to be in $\{x_i\}_{i=1}^j$ and can be an arbitrary point in $\mathcal{X}$.
This permits the development of quite general strategies for SP-MCMC, beyond those explicitly conisdered in the main text.

Armed with the notion of $V$-uniform ergodicity, we derive the following general result: 
\begin{theorem}[SP-MCMC for $V$-Uniformly Ergodic Markov Chains]
\label{thm: V-UE MCMC general}
Suppose $\int_{\mathcal{X}} k_0(x,\cdot) \mathrm{d}P(x) \equiv 0$ and let $(m_j)_{j = 1}^n \subset \mathbb{N}$ be a fixed sequence.
Fix a function $V : \mathcal{X} \rightarrow [1,\infty)$ and consider time-homogeneous reversible Markov chains $(Y_{j,l})_{l = 1}^{m_j}$, $j \in \mathbb{N}$, generated using the same $V$-uniformly ergodic transition kernel.
Suppose $\exists \gamma> 0$ such that $b := \E[e^{\gamma k_0(Y, Y)}] < \infty$.
Let $\{x_i\}_{i=1}^n$ denote the output of SP-MCMC. 
For each $1 \leq i \leq n$, fix $S_i \geq 0$ and $r_i > 0$.
Then for some constant $C$, 
\begin{eqnarray}
\mathbb{E} \left[ D_{\mathcal{K}_0,P}(\{x_i\}_{i=1}^n)^2 \right] \leq C \exp\left( \sum_{i=1}^n \frac{1}{r_i} \right) \frac{1}{n^2} \sum_{i=1}^n \left( S_i^2 + r_i e^{- \frac{\gamma}{2} S_i^2} +  V_+(S_i) V_-(S_i) \frac{r_i}{m_i} \right) \label{eq: complicated V-UE result}
\end{eqnarray}
where in each case the total expectation $\mathbb{E}$ is taken over realisations of the random sets $\mathcal{Y}_j = \{Y_{j,l}\}_{l = 1}^{m_j}$, $j \in \mathbb{N}$.
\end{theorem}
\begin{proof}
\def\indic#1{\mathbb{I}\left[{#1}\right]} 
\def\norm#1{\|{#1}\|} 

The stucture of the proof is initially identical to that used in Theorem \ref{thm: SP ideal MC general result}.
Indeed, proceeding as in the proof of Theorem \ref{thm: SP ideal MC general result} we set up the triangle inequality in \eqref{eq: triangle ineq h trun IID} and attempt to control both terms in this bound.
For the first term we proceed identically to obtain the bound on $\|k_j^-\|_{\mathcal{K}_0}^2$ in \eqref{eq: kjminus bound}.
For the second term we proceed identically to obtain the bound 
\begin{eqnarray}
\mathbb{E} \left[ \|h_j - k_j^-\|_{\mathcal{K}_0}^2 \right] = \frac{1}{m_j^2} \sum_{l=1}^{m_j} \sum_{l'=1}^{m_j} \left\{ \mathbb{E}[ h_{Y_{j,l'}}(Y_{j,l}) ] - \mathbb{E}[ h_{Y_{j,l'}}(Z) ] \right\} - \frac{1}{m_j} \sum_{l = 1}^{m_j} \int \left\{ \mathbb{E}[ h_x(Y_{j,l}) ] - \mathbb{E}[ h_x(Z) ] \right\} \mathrm{d}P(x) \label{eq: truncated embedding error}
\end{eqnarray}
from \eqref{eq: truncated embedding error IID}, where $Z \sim P$ is independent of all other random variables that we have defined.
However, the subsequent argument in Theorem \ref{thm: SP ideal MC general result} exploited independence of the random variables $Y_{j,l}$, which does not hold in the Markov chain context.
Thus our aim in the sequel is to leverage $V$-uniform ergodicity to control \eqref{eq: truncated embedding error}.

For the first term in \eqref{eq: truncated embedding error} we exploit the definition of the $\|\cdot\|_V$ norm and the Cauchy-Schwarz inequality to see that, for each $l' \geq l$, we have that
\begin{eqnarray}
\left| \mathbb{E}[ h_{Y_{j,l'}}(Y_{j,l}) | Y_{j,l'} = y ] - \mathbb{E}[ h_{Y_{j,l'}}(Z) | Y_{j,l'} = y ] \right| & = & \left| \mathbb{E}[ h_{y}(Y_{j,l}) | Y_{j,l'} = y ] - \mathbb{E}[ h_{y}(Z) ] \right| \nonumber \\
& \leq & \|\mathrm{P}^{l'-l}(y,\cdot) - P\|_V \|h_y\|_V \nonumber  \\
& = & \|\mathrm{P}^{l'-l}(y,\cdot) - P\|_V \sup_{x \in \mathcal{X}} \frac{|k_0(y,x) \mathbb{I}[x,y \in B_j]|}{V(x)} \nonumber \\
& \leq & \|\mathrm{P}^{l'-l}(y,\cdot) - P\|_V \sup_{x \in \mathcal{X}} \frac{k_0(x,x)^{\frac{1}{2}} k_0(y,y)^{\frac{1}{2}} \mathbb{I}[x,y \in B_j]}{V(x)} \nonumber \\
& = & \|\mathrm{P}^{l'-l}(y,\cdot) - P\|_V k_0(y,y)^{\frac{1}{2}} \mathbb{I}[y \in B_j] \sup_{x \in B_j} \frac{k_0(x,x)^{\frac{1}{2}}}{V(x)} . \label{eq: several terms}
\end{eqnarray}
At this point we exploit $V$-uniform ergodicity to obtain that, for some $R \in [0,\infty)$ and $\rho \in (0,1)$, 
\begin{eqnarray}
\eqref{eq: several terms} & \leq & R V(y) \rho^{l' - l} \times k_0(y,y)^{\frac{1}{2}} \mathbb{I}[y \in B_j] \times \sup_{x \in B_j} \frac{k_0(x,x)^{\frac{1}{2}}}{V(x)} \nonumber \\
& \leq & R \rho^{l' - l} \times \sup_{y \in B_j} V(y) k_0(y,y)^{\frac{1}{2}} \times \sup_{x \in B_j} \frac{k_0(x,x)^{\frac{1}{2}}}{V(x)} \nonumber \\
& \leq & R \rho^{l'-l} \times V_+(S_j) V_-(S_j). \label{eq: use v-ergodic}
\end{eqnarray}
Note that from the symmetry $h_x(x') = h_{x'}(x)$, together with the fact that the Markov chain is reversible, the above bound holds also for $l' < l$ if $l' - l$ is replaced by $|l' - l|$.
Thus, from Jensen's inequality,
\begin{eqnarray*}
\left| \mathbb{E}[h_{Y_{j,l'}}(Y_{j,l})] - \mathbb{E}[h_{Y_{j,l'}}(Z)] \right| & = & \left| \mathbb{E} \left[ \mathbb{E}[h_{Y_{j,l'}}(Y_{j,l}) | Y_{j,l'}] - \mathbb{E}[h_{Y_{j,l'}}(Z) | Y_{j,l'}] \right] \right|
\\
& \leq & \mathbb{E} \left| \left[ \mathbb{E}[h_{Y_{j,l'}}(Y_{j,l}) | Y_{j,l'}] - \mathbb{E}[h_{Y_{j,l'}}(Z) | Y_{j,l'}] \right] \right|
\\
& \leq & R V_+(S_j) V_-(S_j)\rho^{|l' - l|} 
\end{eqnarray*}
from which it follows that
\begin{eqnarray*}
\left| \frac{1}{m_j^2} \sum_{l=1}^{m_j} \sum_{l'=1}^{m_j} \left\{ \mathbb{E}[ h_{Y_{j,l'}}(Y_{j,l}) ] - \mathbb{E}[ h_{Y_{j,l'}}(Z) ] \right\} \right| & \leq & R V_+(S_j) V_-(S_j) \frac{1}{m_j^2} \sum_{l=1}^{m_j} \sum_{l'=1}^{m_j} \rho^{|l' - l|} \\
& = & R V_+(S_j) V_-(S_j) \left[ \frac{1}{m_j} + \frac{2}{m_j} \sum_{r=1}^{m_j-1} \left( \frac{m_j-r}{m_j} \right) \rho^r \right] \\
& \leq & R V_+(S_j) V_-(S_j) \left[ \frac{1}{m_j} + \frac{2}{m_j} \sum_{r=1}^{\infty} \rho^r \right] \\
& = & R V_+(S_j) V_-(S_j) \left( \frac{1 + \rho}{1 - \rho} \right) \frac{1}{m_j} .
\end{eqnarray*}
For the second term in \eqref{eq: truncated embedding error} we use the same approach as in \eqref{eq: use v-ergodic} to obtain that
$$
\left| \mathbb{E}[h_x(Y_{j,l}) | Y_{j,0} = x_{j-1}] - \mathbb{E}[h_x(Z) | Y_{j,0} = x_{j-1}] \right| \leq R V_+(S_j) V_-(S_j) \rho^l
$$
independently of $x_{j-1}$, and hence that
\begin{eqnarray*}
\left| \frac{1}{m_j} \sum_{l = 1}^{m_j} \int \mathbb{E}[ h_x(Y_{j,l}) ] - \mathbb{E}[ h_x(Z) ] \mathrm{d}P(x) \right| & \leq & R V_+(S_j) V_-(S_j) \frac{1}{m_j} \sum_{l=1}^{m_j} \rho^l \\
& \leq & R V_+(S_j) V_-(S_j) \left( \frac{\rho}{1-\rho} \right) \frac{1}{m_j}.
\end{eqnarray*}

\paragraph{Overall Bound:}

Combining our bounds for the terms in \eqref{eq: triangle ineq h trun IID} leads to
\begin{eqnarray}
\mathbb{E}[ \|h_j\|_{\mathcal{K}_0}^2 ] & \leq & 2 \left( \| k_j^- \|_{\mathcal{K}_0}^2 + \mathbb{E}\left[ \|h_j - k_j^- \|_{\mathcal{K}_0}^2 \right] \right) \nonumber \\
& \leq & \frac{4b}{\gamma} e^{-\frac{\gamma}{2} S_j^2} + 2 R V_+(S_j) V_-(S_j) \left( \frac{1+\rho}{1-\rho} \right) \frac{1}{m_j} + 2 R V_+(S_j) V_-(S_j) \left( \frac{\rho}{1-\rho} \right) \frac{1}{m_j} \nonumber \\
& = & \frac{4b}{\gamma} e^{-\frac{\gamma}{2} S_j^2} + 2 R V_+(S_j) V_-(S_j) \left( \frac{1+2\rho}{1-\rho} \right) \frac{1}{m_j} . \label{eq: EH bound}
\end{eqnarray}
In a similar manner to the proof of Theorem \ref{thm: SP ideal MC general result} we can square \eqref{eq: thm4 result} (with, recall, $\delta = 0$) and take expectations, which combine with \eqref{eq: EH bound} to produce \eqref{eq: complicated V-UE result} with $C = \max\left\{ \frac{4b}{\gamma} , 2R \left( \frac{1+2\rho}{1-\rho} \right) \right\}$.

\end{proof}

Theorem \ref{thm: main theorem} in the main text follows as a special case of the previous result:

\begin{proof}[Proof of Theorem \ref{thm: main theorem}]
The result follows from specialising \eqref{eq: complicated V-UE result} to the case $S_i^2 = \frac{2}{\gamma} \log(n \wedge m_i)$ and $r_i = n$.
Indeed, with these settings \eqref{eq: complicated V-UE result} can be upper-bounded as follows (with $C$ playing the role of a generic constant changing from line to line):
\begin{eqnarray}
\mathbb{E} \left[ D_{\mathcal{K}_0,P}(\{x_i\}_{i=1}^n)^2 \right] & \leq & C \exp\left( \sum_{i=1}^n \frac{1}{r_i} \right) \frac{1}{n^2} \sum_{i=1}^n \left( S_i^2 + r_i e^{- \frac{\gamma}{2} S_i^2} +  V_+(S_i) V_-(S_i) \frac{r_i}{m_i} \right) \nonumber \\
& = & C \frac{1}{n^2} \sum_{i=1}^n \left( \frac{2}{\gamma} \log(n \wedge m_i) + \frac{n}{n \wedge m_i} + V_+(S_i) V_-(S_i) \frac{n}{m_i} \right) \nonumber \\
& \leq & C \frac{1}{n} \sum_{i=1}^n \left( \frac{\log(n \wedge m_i)}{n} + \frac{ V_+(S_i) V_-(S_i)}{m_i} \right)
\end{eqnarray}
as claimed.
\end{proof}

\subsection{Proof of Theorem \ref{thm: SPMCMC consistent}} \label{ap: prf mala consistency}

Our aim is to check the ergodicity preconditions of \cref{thm: main theorem} when the Metropolis-adjusted Langevin algorithm (MALA) transition kernel is employed.
To this end, we will establish $V$-uniform ergodicity of MALA for the specific choice $V(x) = 1 + \|x\|_2$.
This will imply the convergence of SP-MCMC, as motivated by the following result:

\begin{theorem}[SP-MCMC Convergence 2]
\label{thm: SPMCMC conv 2}
Suppose $k_0$ has the form \cref{eq:stein_kernel}, based on a kernel $k \in C_b^{(1,1)}$ and a target $P\in\pset$ such that $\int_{\mathcal{X}} k_0(x,\cdot) \mathrm{d}P(x) \equiv 0$. 
Let $(m_j)_{j = 1}^n \subset \mathbb{N}$ be a fixed sequence.
Consider the function $V(x) = 1 + \|x\|_2$ and consider time-homogeneous reversible Markov chains $(Y_{j,l})_{l = 1}^{m_j}$, $j \in \mathbb{N}$, generated using the same $V$-uniformly ergodic transition kernel.
Let $\{x_i\}_{i=1}^n$ denote the output of SP-MCMC. 
Then $\exists \; C>0$ such that
\begin{eqnarray*}
\mathbb{E}\left[D_{\mathcal{K}_0, P}(\{x_i\}_{i=1}^n)^2\right] \leq C \; \frac{1}{n} \sum_{i=1}^n \frac{\log(n \wedge m_i)}{n \wedge m_i} 
\end{eqnarray*}
where in each case the total expectation $\mathbb{E}$ is taken over realisations of the random sets $\mathcal{Y}_j = \{Y_{j,l}\}_{l=1}^{m_j}$, $j \in \mathbb{N}$.
\end{theorem}
\begin{proof}
Firstly, consider the case where $\exists C_0 > 0$ such that $\frac{1}{C_0} \leq V_0(x) := \frac{V(x)}{1 + k_0(x,x)^{1/2}} \leq C_0$.
In this situation we have that the functions $V_+$ and $V_-$ defined in Theorem \ref{thm: main theorem} satisfy $V_+(s) \leq C_0(s+s^2)$ and $V_-(s) \leq C_0$.
It therefore follows from Theorem \ref{thm: main theorem} that
\begin{eqnarray*}
\mathbb{E} \left[ D_{\mathcal{K}_0,P}(\{x_i\}_{i=1}^n)^2 \right] \; \leq \; C \frac{1}{n}  \sum_{i=1}^n \left( \frac{\log(n \wedge m_i)}{n} + \frac{ V_+(S_i) V_-(S_i)}{m_i} \right) \; \leq \; C \frac{1}{n} \sum_{i=1}^n \frac{\log(n \wedge m_i)}{n \wedge m_i} 
\end{eqnarray*}
again with $C$ a generic constant.
It is therefore sufficient to establish that $\frac{1}{C_0} \leq V_0(x) \leq C_0$ is a consequence of the $V$-uniform ergodicity with $V(x) = 1 + \|x\|_2$ that we have assumed.
The lower and upper bounds on $V_0$ are derived separately in the sequel.

\paragraph{Lower Bound:}

If $\nabla \log p$ is Lipschitz and $k(x,y) = \psi(\|x-y\|_2^2)$ with $\psi \in C^2$ (which is the case for the pre-conditioned IMQ kernel), then it can be shown that $k_0(x,x) \leq B \|x\|_2^2 + D$ for some $B$ and $D$.
Indeed, recall from \eqref{eq:stein_kernel} that
\begin{align*}
k_0(x,y) & \; = \; \nabla_x \cdot \nabla_{y} k(x,y) + \left\langle \nabla_x k(x,y) , \nabla_{y} \log p(y) \right\rangle   + \left\langle \nabla_{y} k(x,y) , \nabla_x \log p(x) \right\rangle  + k(x,y) \left\langle \nabla_{x} \log p(x) , \nabla_{y} \log p(y) \right\rangle
\end{align*}
and note $\nabla_x k(x,y) = 2(x-y) \psi'(\|x-y\|_2^2)$, thus $\nabla_x k(x,x)=0$ and $k_0(x,x) = \psi(0) \| \nabla_{x} \log p(x) \|_2^2 $. 
From Lipschitz continuity of $ \nabla \log p$, say with Lipschitz constant $C$, we have that $\| \nabla \log p(x) \|_2 \leq C \|x\|_2 + \| \nabla \log p(0) \|_2$ . Let $A=\| \nabla \log p(0) \|_2$.
Thus 
$$\| \nabla \log p(x) \|_2^2 \leq C^2 \|x\|_2^2 + 2AC \|x\|_2 +A^2.$$
For $\| x \|_2 \leq 1$ we have that
$ C^2 \|x\|_2^2 + 2AC \|x\|_2 +A^2 \leq  C^2 \|x\|_2^2 + 2AC +A^2$, and for $\| x\|_2 \geq 1$ we have that
$ C^2 \|x\|_2^2 + 2AC \|x\|_2 +A^2 \leq  (C^2+2AC) \|x\|_2^2 +A^2$. Hence for all $x$, $\exists D,B$ with
\begin{align}
k_0(x,x) \leq B \| x \|_2^2 +D \label{eq: quad bound on k0}
\end{align}
as claimed.
This result implies that 
$$
\sup_{x \in \mathcal{X}} V_0(x)^{-1} = \sup_{x \in \mathcal{X}} \frac{1 + k_0(x,x)^{\frac{1}{2}}}{V(x)} \leq \sup_{x \in \mathcal{X}} \frac{1 + \sqrt{B \|x\|_2^2 + D}}{1 + \|x\|_2} < \infty .
$$

\paragraph{Upper Bound:}

For the converse direction we make use of the distant dissipativity assumption.
Recall that this implies
\begin{eqnarray*}
\kappa(\|x-y\|_2) \|x-y\|_2^2 \; \leq \; -2 \langle \nabla \log p(x) - \nabla \log p(y) , x - y \rangle
\end{eqnarray*}
and thus, setting $y = 0$, taking the absolute value on the right hand side and using the Cauchy-Schwarz inequality,
\begin{eqnarray*}
\kappa(\|x\|_2) \|x\|_2^2 \; \leq \; 2 | \langle \nabla \log p(x) - \nabla \log p(0) , x \rangle | \; \leq \; 2 \|\nabla \log p(x) - \nabla \log p(0)\|_2 \|x\|_2 .
\end{eqnarray*}
Rearranging, and using the triangle inequality,
\begin{eqnarray*}
\kappa(\|x\|_2) \|x\|_2 \; \leq \; 2 \|\nabla \log p(x) - \nabla \log p(0)\|_2  \; \leq \; 2 \| \nabla \log p(x) \|_2 + 2 \| \nabla \log p(0) \|_2
\end{eqnarray*}
and rearranging again,
\begin{eqnarray*}
- 2 \| \nabla \log p(0) \|_2 + \kappa(\|x\|_2) \|x\|_2 \; \leq \; 2 \| \nabla \log p(x) \|_2 .
\end{eqnarray*}
Now, since $\kappa_0 = \lim \inf_{r \rightarrow \infty} \kappa(r) > 0$, there $\exists R$ such that, for all $\|x\|_2 > R$, $\kappa(\|x\|_2) \geq \frac{\kappa_0}{2} > 0$ and hence, for $\|x\|_2 > R$,
\begin{eqnarray*}
- 2 \| \nabla \log p(0) \|_2 + \frac{\kappa_0}{2} \|x\|_2 & \leq & 2 \| \nabla \log p(x) \|_2 ,
\end{eqnarray*}
where we are free to additionally assume that 
\begin{eqnarray}
R > 1 + \frac{8}{\kappa_0} \|\nabla \log p(0) \|_2 . \label{eq: R big enough}
\end{eqnarray}
Since $\|x\|_2 > R$, it follows that
\begin{eqnarray*}
- 2 \| \nabla \log p(0) \|_2 + + \frac{\kappa_0}{4} R +  \frac{\kappa_0}{4} \|x\|_2 & \leq & 2 \| \nabla \log p(x) \|_2 
\end{eqnarray*}
and from \eqref{eq: R big enough} we further deduce that
\begin{eqnarray*}
1 + \|x\|_2 & \leq & \frac{8}{\kappa_0} \| \nabla \log p(x) \|_2 
\end{eqnarray*}
for all $\|x\|_2 > R$. 
Thus
\begin{eqnarray*}
\sup_{x \in \mathcal{X}} V_0(x) \; = \; \sup_{x \in \mathcal{X}} \frac{V(x)}{1 + k_0(x,x)^{\frac{1}{2}}} & \leq & \sup_{\|x\|_2 \leq R} \frac{1 + \|x\|_2}{1 + k_0(x,x)^{\frac{1}{2}}} + \sup_{\|x\|_2 > R} \frac{1 + \|x\|_2}{1 + k_0(x,x)^{\frac{1}{2}}} \\
& \leq & (1 + R) + \sup_{\|x\|_2 > R} \frac{\frac{8}{\kappa_0} \|\nabla \log p(x) \|_2}{1 + \psi(0)^{1/2} \|\nabla \log p(x) \|_2} \; < \; \infty
\end{eqnarray*}
as required.
\end{proof}

The implication of Theorem \ref{thm: SPMCMC conv 2} is that we can seek to establish $V$-uniform ergodicity of MALA in the case $V(x) = 1 + \|x\|_2$.
To this end, we present Lemmas \ref{lem: Equivalent Ergodicity}, \ref{lem: SE ergo ULA}, Proposition \ref{prop: SE ergo ULA} and Theorem \ref{thrm-squared-exp-ergod} next:

\begin{lemma}\label{lem: Equivalent Ergodicity} Let $U, V: \mathcal X \to [1,\infty)$ be functions such that $V<cU$ and $U<aV$ for $c,a >0$. Then a Markov chain is $U$-uniformly ergodic if and only if it is $V$-uniformly ergodic.
\end{lemma}
\begin{proof}
Suppose $\exists R_U \in [0,\infty), \rho_U \in (0,1)$ such that 
$$ \| \mathrm{P}^n(y,\cdot)-P \|_{U} \leq R_U U(y) \rho_U^n$$ for all initial states $y \in \mathcal{X}$.
From the definition of the $V$-norm, we have that $\frac1 c \| f \|_U \leq \| f \|_V \leq a \| f \|_U$, and moreover
$$
\|\mu \|_V = \sup_{\| f\|_V \neq 0} \frac{| \mu f |}{\|f \|_V} \leq  c  \|\mu \|_U.
$$
Together, these imply that 
$$
\| \mathrm{P}^n(y,\cdot)-P \|_{V} \leq  c \| \mathrm{P}^n(y,\cdot)-P \|_{U} \leq c R_U U(y) \rho_U^n \leq  acR_U V(y) \rho_U^n = R_V V(y) \rho_V^n
$$
where $R_V := a c R_U \in [0,\infty)$ and $\rho_V = \rho_U \in (0,1)$.
Thus $U$-uniform ergodicity implies $V$-uniform ergodicity. 
The converse result follows by symmetry.
\end{proof}

The next Lemma concerns the unadjusted Langevin algorithm (ULA), whose proposal distribution is identical to MALA, but the acceptance/rejection step is not perfomed \citep{Roberts1996}.
As such, ULA does not leave $P$ invariant but leaves a different distribution, which we denote $\tilde{P}$, invariant.

\begin{lemma}[Properties of ULA Proposal Distribution] 
\label{lem: SE ergo ULA}
Suppose $P \in \pset$ and $\mathcal{X} = \mathbb{R}^d$. 
Let $c(x) := x + \frac h 2 \nabla \log p(x)$.
Then $\exists R, h_0, \kappa > 0$ such that, for $\|x\|_2 > R$, it holds that
\begin{align}
\|c(x)\|_2^2 < \Big( 1 - \frac{\kappa h}{2} \Big) \|x\|_2^2 \label{eq: cx bound}
\end{align}
whenever the step size $h$ satisfies $h < h_0$.
Moreover let $Y$ be distributed as the MALA proposal distribution starting from $x \in \mathcal{X}$, namely $Y \stackrel{d}{=} c(x) + \sqrt{h}Z$, where $h > 0$ and $Z \sim \mathcal{N}(0,I)$ where $I$ is a $d \times d$ identity matrix.
Then $\exists R_2, h_0, \kappa_2 > 0$ such that for $\|x\|_2 > R_2$, $s < 1 / 2h$, it holds that
$$
\E\left[ \exp\left(s \| Y \|_2^2\right) \right] 
\; < \;
 \frac{1}{(1-2sh)^{d/2}} \exp\left( \left[ \frac{1 - \frac{\kappa_2}{2} }{1 - 2 s h} \right] s \| x \|_2^2 \right) 
$$
whenever the step size $h$ satisfies $h < h_0$.
Let $A(h) := \kappa_2h/(4 - \kappa_2 h)$.
Then, furthermore, if $s < \min \left\{ 1/2h, \kappa_2/8 \right\}$, then
$$
\E\left[ \exp\left(s \| Y \|_2^2\right) \right] 
\; < \;
 \frac{4^{d/2}}{(4-\kappa_2 h)^{d/2}} \exp\left([1 - A(h)] s \| x \|_2^2 \right) 
$$
and $A(h) \in (0,1)$ whenever the step size $h$ satisfies $h < \min \left\{ h_0 , 2/\kappa_2 \right\}$.
\end{lemma}

\begin{proof}

We expand 
\begin{align}
\|c(x) \|_2^2
\; = \;
 \left\langle x+ \frac{h}{2}\nabla \log p(x),x+\frac{h}{2}\nabla \log p(x) \right\rangle
\; = \;
 \| x\|_2^2 +h \left\langle x,\nabla \log p(x) \right\rangle+\frac{h^2}{4} \left\| \nabla \log p(x)\right\|_2^2 . \label{eq: c expand}
\end{align}
Since $P$ is distantly dissipative, we know for any $r > 0$ and any $x$ with $\| x \|_2 =r$, 
$$
\kappa(r)
 \; \leq \;
 -2 \frac{\langle \nabla \log p(x)-\nabla \log p(0),x \rangle}{\| x\|_2^2}.
$$ 
Moreover, since $\kappa_0 :=\lim \inf_{r \rightarrow \infty} \kappa(r) >0$, $\exists R_1 > 0$ and $\kappa_1 \in (0,\kappa_0)$ such that for all $\|x\|_2 > R_1$, we have $\kappa(\|x\|_2) \geq \kappa_1$.
Thus for any $\|x\|_2 > R_1$ we have that
\begin{align*}
\kappa_1 \| x\|_2^2 
& \; \leq \; -2 \langle \nabla \log p(x)-\nabla \log p(0),x \rangle \\
& \; = \; - \langle \nabla \log p(x),x \rangle + \langle \nabla \log p(0),x \rangle \leq - \langle \nabla \log p(x),x \rangle + \| \nabla \log p(0)\|_2 \|x \|_2   
\end{align*}
Let $\kappa_2 \in (0, \kappa_1)$, so that for $\| x \|_2 > \| \nabla \log p(0) \|_2 /\kappa_1 -\kappa_2$, we have $\kappa_2 \| x \|^2_2 \leq \kappa_1 \| x \|_2^2 - \| \nabla \log p(0) \| \|x \|_2$. 
Hence, for $\| x \|_2 > R_2 := \max\left\{ R_1 , \| \nabla \log p(0) \|_2 /\kappa_1 -\kappa_2 \right\}$, we have that 
$$ 
\kappa_2 \| x \|^2 
\; \leq \; - \langle \nabla \log p(x),x \rangle .
$$
Since $\nabla \log p$ is Lipschitz, there exists a constant $L$ such that for all $x$ we have $\| \nabla \log p(x) - \nabla \log p(0) \|_2 \leq L \|x\|_2$.
It follows that for all $\|x\|_2 > R_2$,
$$
\| \nabla \log p(x) \|_2 
\; \leq \;
\underbrace{\left( L + \frac{\|\nabla \log p(0) \|_2}{R_2} \right)}_{=: \tilde{L}} \|x\|_2 .
$$
Hence for $\| x \|_2 > R_2$ we have, from \eqref{eq: c expand} and the bounds just obtained, that $\|c(x)\|_2^2 \leq  \| x \|_2^2 - h \kappa_2 \| x\|_2^2 + (h^2/4) \tilde{L}^2 \|x \|_2^2$. 
For $h < h_0 := 2 \kappa_2 / \tilde{L}^2$ we have that $h \kappa_2)/2 > (h^2/4) \tilde{L}^2 $, so that $1 - h \kappa_2 + \frac{h^2}{4} \tilde{L}^2 < 1 - (\kappa_2 h)/2$ and therefore have that $\|c(x) \|_2^2 < \left( 1 - \kappa_2 h/2 \right) \| x \|_2^2$. 
The first part of the Lemma is now established.

For the second part,
In this theorem, we consider the proposal distribution of MALA which is an unadjusted Langevin algorithm (ULA). First note that
\begin{align}
\E\left[ \exp\left(s \| Y \|_2^2\right) \right] 
\; = \;
 \E \left[ \exp\left(s \left\| \sqrt{h}Z+c(x) \right\|_2^2\right) \right] 
 \; = \;
  \E \left[ \exp\left(sh \left\| Z+ \frac{c(x)}{\sqrt{h}}\right \|_2^2\right) \right]
 \; = \;
  \E \left[ \exp\left( sh W^2 \right) \right] , \label{eq: Y and W RVs}
\end{align}
where $W := \| Z+ c(x)/\sqrt{h} \|_2^2$ is a non-central chi-squared random variable with non-centrality parameter $\lambda = \|c(x)\|_2^2/h$ and degrees of freedom $d$. 
The last expression is recognised as the moment generating function $M_W(t) = \mathbb{E}[e^{tW}]$ of $W$, evaluated at $t = sh$. 
Recall that $M_W(t) = (1/(1-2t)^{d/2}) \exp(\lambda t/1 - 2t)$, valid for $2t < 1$ \citep[Sec. 26.4.25 of][]{Abramowitz1972}.
Hence, for $2sh < 1$ we have that
$$ 
\E\left[ \exp\left(s \| Y \|_2^2\right) \right]
\; = \; 
\frac{1}{(1-2sh)^{d/2}} \exp\left (\frac{s \| c(x) \|_2^2}{1 - 2 s h} \right) .
$$

We then observe that if additionally $s < \kappa_2 / 8$ and $h < 2/ \kappa_2$ then 
$$
\frac{1 - \frac{\kappa_2}{2}h}{1 - 2 s h} 
\; < \; 
\frac{1 - \frac{\kappa_2}{2}h}{1 - \frac{\kappa_2}{4} h} = 1 - A(h),
 \qquad
 \frac{1}{(1 - 2 s h)^{d/2}}
   \; < \;
 \frac{4^{d/2}}{(4 - \kappa_2 h)^{d/2}}, 
 \qquad  \text{and} \qquad 
 A(h) \; = \; \frac{\kappa_2 h}{4 - \kappa_2 h} < 1
$$
as required.
\end{proof}

\begin{proposition}[$U_s$-Uniform Ergodicity of ULA] 
\label{prop: SE ergo ULA}
Suppose $P \in \pset$ and $\mathcal{X} = \mathbb{R}^d$. 
The one-step transition kernel $\mathrm{P} = \mathrm{P}^1$ of ULA satisfies
\begin{align}\label{eq:ula-drift}
\mathrm{P}U_s(x) \leq \tau_1 U_s(x) + \tau_0
\end{align}
for some $\tau_1 < 1$ and $\tau_0 \in \mathbb{R}$, for each of $U_s(x) = \exp(s\|x\|_2)$ (any $s > 0$), $\exp(s\|x\|_2^2)$ (some $s > 0$), and $U_s(x) = 1 + \|x\|_2^s$ ($s \in \{1,2\}$).
Thus ULA is $U_s$-uniformly ergodic for its invariant distribution $\tilde{P}$ for each of these $U_s(x)$.
\end{proposition}
\begin{proof}
The strategy of the proof is to establish the geometric drift condition
\begin{align}
\limsup_{\|x\|_2 \rightarrow \infty} \frac{ \mathrm{P} U_s(x)}{U_s(x)} < 1 \label{eq: geom drift cond}
\end{align}
for each of the functions $U_s$ given in the statement.
Let $c(x) := x + \frac h 2 \nabla \log p(x)$ and consider the ULA density at a given point $x$, defined as 
$$
q(x,y) := \frac{1}{(2\pi h)^{k/2}} \exp\Big( -\frac 1 {2h}\|y-c(x)\|_2^2\Big).
$$
Then we have that
\begin{align*}
\gamma(x) & \; := \; U_s(x)^{-1} \int_{\mathcal{X}} q(x,y) U_s(y) \mathrm{d}y\\
 & \;= \;
 \frac{U_s(x)^{-1}}{(2\pi h)^{k/2}} \int  \exp\Big( -\frac 1 {2h}\|y-c(x)\|_2^2\Big) U_s(y) \mathrm{d}y \\
& \; = \;
 \frac{1}{(2\pi h)^{k/2}} U_s(x)^{-1}\int  \exp\Big( -\frac 1 {2h}\|y\|_2^2\Big) U_s \big(y+c(x)\big) \mathrm{d}y \\
& \; = \;
 \frac{1}{(2\pi)^{k/2}}U_s(x)^{-1}\int  \exp\Big( -\frac 1 2 \|y\|_2^2\Big) U_s \big(\sqrt{h}y+c(x) \big) \mathrm{d}y 
\; = \;
 U_s(x)^{-1} \E\left[ U_s(Y) \right] 
\end{align*}
where $Y$ is the random variable defined in the statement of \cref{lem: SE ergo ULA}; c.f. \eqref{eq: Y and W RVs}.
Now we consider each of the functions $U_s(x) = \exp(s\|x\|_2)$, $\exp(s\|x\|_2^2)$ and $1 + \|x\|_2^2$ in turn (the case $1 + \|x\|_2$ will be treated separately at the end):
\begin{itemize}

\item For $U_s(x) = \exp(s\|x\|_2)$, let $\tilde{Z} = h^{\frac{1}{2}}Z$ where $Z \sim \mathcal{N}(0,I)$, so that
\begin{align}
\gamma(x) = \frac{\E[\exp(s\|Y\|_2)]}{\exp(s\|x\|_2)} = \frac{\E[\exp(s\|\tilde{Z}+c(x)\|_2)]}{\exp(s\|x\|_2)} &\leq \frac{\E[\exp(s\|\tilde{Z}\|_2)] \exp(s\|c(x)\|_2)}{\exp(s\|x\|_2)} \label{eq: s1 triangle ineq} \\
& \leq \E[\exp(s\|\tilde{Z}\|_2)] \exp \Big( \Big( \sqrt{1- \frac{\kappa h}{2}} - 1 \Big) s\|x\|_2 \Big) \label{eq: finish s1 case}
\end{align}
where we have used the triangle inequality in \eqref{eq: s1 triangle ineq} and we have used \eqref{eq: cx bound} from \cref{lem: SE ergo ULA} to obtain \eqref{eq: finish s1 case}.
The final bound goes to zero as $\|x\|_2 \rightarrow \infty$, since a Gaussian has finite exponential moments, and thus the geometric drift condition is satisfied.

\item For $U_s(x) = \exp( s \| x \|_2^2  )$, from the conclusion of Lemma \ref{lem: SE ergo ULA} we have that, with $A(h) := \kappa_2h/(4 - \kappa_2 h)$,
\begin{align*}
\gamma(x) \; < \; \frac{4^{d/2}}{(4-\kappa_2 h)^{d/2}} \exp\left(-s\|x\|_2^2\right) \exp\left([1 - A(h)] s \| x \|_2^2 \right)  \; = \; \frac{4^{d/2}}{(4-\kappa_2 h)^{d/2}} \exp\left(- A(h) s \| x \|_2^2 \right) 
\end{align*}
where $A(h) \in (0,1)$.
It is therefore clear that for $\|x\|_2$ sufficiently large we have $\gamma(x) < 1$, so that the geometric drift condition is satisfied.

\item For $U_s= 1+ \| x \|^2$, we have 
 \begin{align}
 \gamma(x) \;= \; \frac{ 1+\mathbb E \big[ \big\|Z\sqrt{h}+c(x) \big\|_2^2\big]}{(1+\|x\|_2^2)} \;= \;	
      \frac{ 1+h\mathbb E \big[ \big\|Z\|_2^2\big] +\|c(x)\|_2^2 + 2\sqrt h \mathbb E\big[\langle Z, c(x) \rangle\big] }{(1+\|x\|_2^2)} \nonumber
\end{align}
where $Z \sim \mathcal{N}(0,\mathrm{I})$.
Moreover $\mathbb E\big[\langle Z, c(x) \rangle\big] = \left\langle\mathbb E\big[ Z\big], c(x) \right\rangle =0$, and from Lemma \ref{lem: SE ergo ULA}, $\exists R, h_0, \kappa > 0$ such that for $\|x\|_2 > R$, it holds that
$
 \|c(x)\|_2^2 < \big( 1 - \frac{\kappa h}{2} \big) \| x \|_2^2
$ for $h< h_0$.
Hence
$$ \gamma(x) <\frac{ 1+h\mathbb E \big[ \|Z\|_2^2\big] +\big( 1 - \frac{\kappa h}{2} \big) \| x \|_2^2 }{1+\|x\|_2^2} $$
and for $\|x\|_2$ sufficiently large we have $\gamma(x) < 1$, and the geometric drift condition is satisfied.

\end{itemize}

Thus for $\| x\|>R$ and appropriate $h,s$, there exists $\tau_1 \in (0,1)$ s.t., 
$ \mathrm{P} U_s(x) \leq \tau_1 U_s(x)$, and since $\mathrm{P}U_s$ bounded on the compact set $C=\{ x \in \R^d : \| x \|_2 \leq R\}$, there exists $\tau_0 \in \R$ s.t.,
$ \mathrm{P} U_s(x) \leq \tau_1 U_s(x)+ \tau_0 \mathbb I_C(x)$   for all $x$,
where $\mathbb I$ is the indicator function.
Thus by section 3.1 of \cite{Roberts1996}, the chain is $U_s$-uniformly ergodic for each of $U_s(x) = \exp(s \|x\|_2^2)$, $\exp(s\|x\|_2)$ and $1 + \|x\|_2^2$. 

The remaining case to establish is $U_s$-uniform ergodicity for $U_s(x) = 1 + \|x\|_2^s$ and $s = 1$.
For this, we leverage the fact that $U$-uniform ergodicity implies $\sqrt{U}$-uniform ergodicity by Lemma 15.2.9 \cite{Meyn2012}.
The stated result will then follow from Lemma \ref{lem: Equivalent Ergodicity}, since for some $c > 0$, $\frac{1}{c} U_1(x) \leq \sqrt{U_2(x)} \leq c U_1(x)$.
\end{proof}

Our theoretical analysis now focuses on the MALA transition kernel, which is precisely defined in Appendix \ref{ap: mala explain}.
In what follows, as in the main text, let $q(x,\cdot)$ be a density for the proposal distribution of MALA, starting from the state $x$, and let 
$$
\alpha(x,y) \; := \; \min\left\{ 1 , \frac{p(y) q(y,x)}{p(x) q(x,y)} \right\} 
$$ 
denote the MALA acceptance probability for moving from $x$ to $y$, given that $y$ has been proposed.
As in the main text, we let $A(x) = \{y \in \mathcal{X} : \alpha(x,y) = 1 \}$ denote the region where proposals are always accepted and let $R(x) = \mathcal{X} \setminus A(x)$.
Let $I(x) := \{y : \|y\|_2 \leq \|x\|_2 \}$ represent the set of points interior to $x$.

\begin{theorem}[$V$-Uniform Ergodicity of MALA] 
\label{thrm-squared-exp-ergod}
Suppose $P \in \pset$ and $\mathcal{X} = \mathbb{R}^d$.
Consider MALA with step size $h$ and one-step transition kernel $\mathrm{P} = \mathrm{P}^1$.
Further assume $P$ is such that MALA is inwardly convergent.
Then, for $V = U_s$, where $U_s$ is any of the functions defined in Proposition \ref{prop: SE ergo ULA} for which ULA is $U_s$-uniformly ergodic,
\begin{align}\label{eq:mala-drift}
\mathrm{P}V(x) \leq \tau_1 V(x) + \tau_0.
\end{align}
Hence, in particular, MALA is $V$-uniformly ergodic.
\end{theorem} 
\begin{proof}
The proof strategy follows Theorem 4.1 of \cite{Roberts1996}, which is based on establishing the geometric drift condition \eqref{eq:mala-drift} in the form \eqref{eq: geom drift cond}. 
Let $c(x) := x + \frac h 2 \nabla \log p(x)$ denote the MALA drift and let 
$$
q(x,y) := \frac{1}{(2\pi h)^{k/2}} \exp\Big( -\frac 1 {2h}\|y-c(x)\|_2^2\Big).
$$
Then, using $\mathbb{I}[\cdot]$ to denote the indicator function, the ratio in the geometric drift condition for $V$ can be decomposed and bounded as follows:
\begin{align*}
\frac{\int V(y) \mathrm{P}(x,y) \mathrm{d}y}{V(x)} 
& \; = \;
 \int_{A(x)}  q(x,y) \frac{V(y)}{V(x)}  \mathrm{d}y +  \int_{R(x)} q(x,y) \frac{V(y)}{V(x)}  \alpha(x,y) \mathrm{d}y + \int_{R(x)} q(x,y) [1 - \alpha(x,y)] \mathrm{d}y \\
& \; = \;
 \int_{\mathcal X}  q(x,y) \frac{V(y)}{V(x)}  \mathrm{d}y -  \int_{R(x)} q(x,y) \frac{V(y)}{V(x)} \big(\alpha(x,y) -1 \big) \mathrm{d}y + \int_{R(x)} q(x,y) [1 - \alpha(x,y)] \mathrm{d}y \\
&\; = \;
 \int_{\mathcal{X}} q(x,y) \frac{V(y)}{V(x)}  \mathrm{d}y + \int_{R(x)} q(x,y) \left[ 1 -  \frac{V(y)}{V(x)} \right] [1 - \alpha(x,y)] \mathrm{d}y \\
& \; \leq \;
 \int_{\mathcal{X}} q(x,y) \frac{V(y)}{V(x)}  \mathrm{d}y + \int_{R(x)} q(x,y) \mathbb{I}\left[ 1 - \frac{V(y)}{V(x)}  \geq 0 \right] [1 - \alpha(x,y)] \mathrm{d}y \\
& \;\leq \;
 \int_{\mathcal{X}} q(x,y) \frac{V(y)}{V(x)}  \mathrm{d}y + \int_{R(x) \cap I(x)} q(x,y) \mathrm{d}y 
\end{align*}
where we have used $V(y) \leq V(x)$ for $x \in I(x)$. 
The final term vanishes as $\|x\|_2 \rightarrow \infty$ from the assumption that MALA is inwardly convergent; c.f. \eqref{asm: inw conv}.
So to establish the geometric drift condition for $V$ it remains to show that the first term is asymptotically $<1$, that is ULA is $V$-uniformly ergodic. This was proved in Proposition \ref{prop: SE ergo ULA}.
\end{proof}

Our main result, Theorem \ref{thm: SPMCMC consistent}, follows immediately as a consequence of the results just established and the auxiliary Lemma \ref{lem: sub-Gaussian}:

\begin{proof}[Proof of Theorem \ref{thm: SPMCMC consistent}]
It will be demonstrated that the preconditions of Theorem \ref{thm: SPMCMC conv 2} are satisfied.
Indeed, from Theorem \ref{thrm-squared-exp-ergod} we have that (under our assumptions) MALA is $V$-uniformly ergodic for $V(x) = 1 + \|x\|_2$.
In addition, since $k_0$ has the form \cref{eq:stein_kernel}, based on a kernel $k \in C_b^{(1,1)}$, from \eqref{eq: quad bound on k0} we have that $k_0(x,x) \leq B \| x \|_2^2 +D$ and thus, for $\gamma > 0$ sufficiently small,
$$
\mathbb{E}_{Z \sim P}[ e^{\gamma k_0(Z,Z)}] \leq e^D \mathbb{E}_{Z \sim P}[ e^{\gamma B \|Z\|_2^2}] < \infty
$$
since distant dissipativity of $P$ implies that $P$ is sub-Gaussian (c.f. Lemma \ref{lem: sub-Gaussian}).
It follows that the preconditions of Theorem \ref{thm: SPMCMC conv 2} hold and thus the result is established.
\end{proof}

\begin{lemma} \label{lem: sub-Gaussian}
If $P$ is a distantly dissipative distribution on $\mathcal{X} = \mathbb{R}^d$ with $b(x) := \nabla \log p(x)$ continuous, then $P$ is sub-Gaussian; i.e. $\mathbb{E}_{X \sim P}[e^{a\|X\|_2^2}] < \infty$ for some $a > 0$.
\end{lemma}
\begin{proof}
Since $P$ is distantly dissipative, $\exists R, \kappa$ such that $\langle b(x) - b(y), x - y \rangle \leq -\frac{\kappa}{2} \left\| x-y \right\|_2^2$ holds for all $\|x - y\|_2 \geq R$.
Fix $x \in \mathcal{X}$ with $\|x\|_2 \geq R$ and define $\tau := R / \|x\|_2$.
By the gradient theorem (i.e. the fundamental theorem of calculus for line integrals) applied to the curve $r : [0,1] \rightarrow \mathcal{X}$, $r(t) := tx$, we have that
\begin{align*}
\log p(x) - \log p(0) & = \int_0^1 \langle b(r(t)), r’(t) \rangle \mathrm{d}t \\
& = \int_0^1 \langle b(tx), x \rangle \mathrm{d}t \\
& = \langle b(0),x \rangle + \int_0^1 \frac{1}{t} \langle b(tx) - b(0), tx \rangle \mathrm{d}t \\
& \leq \|b(0)\|_2 \|x\|_2 + \underbrace{\int_0^\tau \langle b(tx) - b(0), x \rangle \mathrm{d}t}_{(*)} + \underbrace{\int_\tau^1 \frac{1}{t} \times - \frac{\kappa}{2} \|tx\|_2^2 \mathrm{d}t}_{(**)}
\end{align*}
where in the final inequality we have used Cauchy-Schwartz followed by the distant dissipativity property of $P$.

The term $(*)$ is an integral of the continuous function $b$ inside the ball $B(0,R)$ and thus, using Cauchy-Schwarz, $(*)$ can be upper bounded by $c_0 \|x\|_2$ for some constant $c_0$ independent of $x$.
The term $(**)$ can be directly evaluated to see that
\begin{align*}
(**) = - \frac{\kappa}{4} \left( 1 - \frac{R^2}{\|x\|_2^2} \right) \|x\|_2^2.
\end{align*}
Thus, for $\|x\|_2 \geq R$, we have the overall bound
\begin{align}
\log p(x) - \log p(0) \leq (\|b(0)\|_2 + c_0) \|x\|_2 - \frac{\kappa}{4}  \left( 1 - \frac{R^2}{\|x\|_2^2} \right) \|x\|_2^2 . \label{eq: overall bound for SG}
\end{align}
A straightforward argument based on \eqref{eq: overall bound for SG} establishes that $\mathbb{E}_{X \sim P}[e^{a\|X\|_2^2}] < \infty$ for some $a>0$, as required.
\end{proof}

\subsection{Proof of Theorem \ref{thm:pre-conditioning}} \label{ap: prf sec 3}

Our final theoretical contribution is to demonstrate that a preconditioned
kernel still ensures the KSD provides control over weak convergence. We
will first state and prove the following useful lemma.

\begin{lemma}[Preconditioned Kernels Maintain Weak Convergence] \label{lem:preconditioned-kernels}
Suppose $P\in\pset$ and let $k(x,y) = \Phi(x-y)$ be some translation invariant
kernel. Fix any positive definite matrix $\Gamma\succ 0$, and let $k_{\Gamma}(x,y) =
\Phi(\Gamma(x-y))$ for all $x,y\in\R^d$. If the KSD based on $k$
controls weak convergence, then the KSD based on $k_{\Gamma}$
also controls weak convergence.
\end{lemma}

\begin{proof}
For any distribution $P$, let us denote
$P_{\Gamma}$ as the distribution of $\Gamma^{-1} Z$ where $Z\sim P$.
We will show that $P_{\Gamma}\in\pset$, and
then the result with follow by making a global change of coordinates
$x\mapsto \Gamma^{-1} x$ and noticing $\textstyle\frac{1}{n}\sum_{i=1}^n
\delta_{x_i}$ converges weakly to $P$ iff $\textstyle\frac{1}{n}\sum_{i=1}^n
\delta_{\Gamma^{-1} x_i}$ converges weakly to $P_{\Gamma}$.

Let $p_{\Gamma}$ and $b_{\Gamma} = \grad\log p_{\Gamma}$ denote the
density and score function of $P_{\Gamma}$, respectively. By a change of
variables we have $p_{\Gamma}(x) = \text{det}(\Gamma) p(\Gamma x)$ and
$b_{\Gamma}(x) = \Gamma \grad\log p(\Gamma x)$ for all $x$. Thus
$b_{\Gamma}$ is Lipschitz since $b = \grad\log p$ is Lipschitz by assumption.

Let $\lambda_{\min}(A)$ and $\lambda_{\max}(A)$ be the
smallest and largest eigenvalue of a positive definite matrix $A$. Then in
the case $\tilde{x} = \Gamma x$ and $\tilde{y} = \Gamma y$, we have
\begin{equation}\label{eqn:distant-dissipative-bound}
\frac{\langle b_{\Gamma}(x) - b_{\Gamma}(y), x - y\rangle}{\|x-y\|_2^2}
  = \frac{\langle b(\Gamma x) - b(\Gamma y), \Gamma (x -
    y)\rangle}{\|x-y\|_2^2}
  = \frac{\langle b(\tilde{x}) - b(\tilde{y}), \tilde{x} -
    \tilde{y}\rangle}{\|\Gamma^{-1}(\tilde{x}-\tilde{y})\|_2^2}
   \le \lambda_{\min}(\Gamma)^{2} \frac{\langle b(\tilde{x}) - b(\tilde{y}), \tilde{x} -
    \tilde{y}\rangle}{\|\tilde{x}-\tilde{y}\|_2^2}
\end{equation}
whenever $\langle b(\tilde{x}) - b(\tilde{y}), \tilde{x} - \tilde{y}\rangle
\le 0$.
Let $\kappa(r) =
\inf\{-2 \frac{\langle b(x)-b(y), x-y\rangle}{\|x-y\|_2^2} :
\|x-y\|_2=r\}$ and $\kappa_{\Gamma}$ be the analogous version using
$b_{\Gamma}$ in lieu of $b$. Combining the fact that
$\lambda_{\max}(\Gamma)^{-1}\|z\|_2 \le \|\Gamma^{-1} z\|_2 \le \lambda_{\min}(\Gamma)^{-1}\|z\|_2$
for all $z$ and
(\ref{eqn:distant-dissipative-bound}) yields
$\sup_{\lambda\in [\lambda_{\max}(\Gamma)^{-1}, \lambda_{\min}(\Gamma)^{-1}]}
\kappa_{\Gamma}(\lambda r) \ge \lambda_{\min}(\Gamma)^{2} \kappa(r)$ for all $r$ such that $\kappa(r) \ge 0$.
Since $P$ is distantly dissipative by assumption, we have
$\lim\inf_{r\to\infty} \kappa(r) > 0$, which implies $\lim\inf_{r\to\infty}
\kappa_{\Gamma}(r) > 0$, and so $P_{\Gamma}\in\pset$ completing the lemma.
\end{proof}

\begin{proof}[Proof of Theorem \ref{thm:pre-conditioning}]
The proof follows by combining Lemma \ref{lem:preconditioned-kernels} with
$\Gamma = \Lambda^{-1/2}$ and Theorem 8 of \cite{Gorham2017}.

\end{proof}



\subsection{Details of Markov Kernels Used}
\label{ap: mala explain}

All experiments in this paper based on MCMC were conducted using either the random walk Metropolis (RWM) algorithm or the Metropolis-adjusted Langevin algorithm (MALA) of \cite{Roberts1996}.
Note that the $P$-invariance of the Markov kernel is not fundamental in SP-MCMC and one could, for example, consider also using an \emph{unadjusted} version in which a Metropolis-Hastings correction is not applied.
However, in order to control for the performance of different MCMC kernels, all experiments in this paper used either the RWM or the MALA kernel, which are both $P$-invariant.

The RWM algorithm is a Metropolis-Hastings method \cite{Metropolis1953} based on the proposal
$$
x \mapsto x + h^{1/2} \xi
$$
where $\xi \sim \mathcal{N}(0,\Sigma)$.
The MALA algorithm is a Metropolis-Hastings method \cite{Metropolis1953} based on the proposal
$$
x \mapsto x + \frac{h}{2} \Sigma^{-1} \nabla \log p(x) + h^{1/2} \xi
$$
where $\xi \sim \mathcal{N}(0,\Sigma)$.
In both cases the \emph{step size} parameter $h$ was calibrated according to the recommendations in \cite{Roberts2001}.
The positive definite matrix $\Sigma$ was taken to be a sample-based approximation to the covariance matrix of $P$, generated by running a long MCMC.\footnote{Our interest in this work was not on the construction of efficient Markov transition kernels, and we defer a more detailed empirical investigation of the impact of poor choice of transition kernel to further work.}

\subsection{Additional Empirical Results} \label{ap: extra results}

\subsubsection{Benchmark Methods} \label{ap: benchmark methods}

In this section we recall the MCMC, MED and SVGD methods used as an empirical benchmark. 
A relatively default version of each method was used.
However, we note that both MED and SVGD are being actively developed and we provide references to more sophisticated formulations of those methods where appropriate in the sequel.

\paragraph{Markov Chain Monte Carlo}

The standard MCMC benchmark in this work was based on a single sample path of MALA, described in Appendix \ref{ap: mala explain}, which is subsequently thinned by discarding all but every $m_j$th point.
Thus the length of the sample path is $m_jn$, where $n$ is the number of points that constitute the final point set.
The choice to keep every $m_j$th state serves to ensure that the MCMC benchmark and SP-MCMC are based on the same Markov kernel, so that empirical results are not confounded.

\paragraph{Minimum Energy Designs} 

The MED method that we consider in this work was proposed as an algorithm for Bayesian computation in \cite{Joseph2015}.
That work restricted attention to $\mathcal{X} = [0,1]^d$ and constructed an energy functional
\begin{eqnarray*}
\mathcal{E}_{\delta,P}(\{x_i\}_{i=1}^n) & := & \sum_{i \neq j} \left[ \frac{ \tilde{p}(x_i)^{-\frac{1}{2d}} \tilde{p}(x_j)^{-\frac{1}{2d}} }{\|x_i - x_j\|_2 } \right]^\delta 
\end{eqnarray*}
for some tuning parameter $\delta \in [1,\infty)$ to be specified.
In \cite{Joseph2017} the default choice of $\delta \rightarrow \infty$ was proposed, so that the energy functional can be interpreted as (up to an appropriate re-normalisation)
\begin{eqnarray*}
\mathcal{E}_{\infty,P}(\{x_i\}_{i=1}^n) & = & \min_{i \neq j} \left[ \frac{ \tilde{p}(x_i)^{-\frac{1}{2d}} \tilde{p}(x_j)^{-\frac{1}{2d}} }{\|x_i - x_j\|_2 } \right] .
\end{eqnarray*}
This default form has the advantage of removing dependence on the hyper-parameter $\delta$ and simultaneously enabling more stable computation, being based on $\log \tilde{p}$ rather than $\tilde{p}$:
\begin{eqnarray*}
\log \mathcal{E}_{\infty,P}(\{x_i\}_{i=1}^n) & = & - \min_{i \neq j} \left[ \frac{1}{2d} \log \tilde{p}(x_i) + \frac{1}{2d} \log \tilde{p}(x_j) + \log \|x_i - x_j\|_2 \right] .
\end{eqnarray*}
A preliminary theoretical analysis of the MED method was also provided in \cite{Joseph2017}.
This focussed on the properties of a point set that globally minimised the energy functional, but did not account for the practical aspect of approximating such a point set.
In fact, minimisation of $\mathcal{E}_{\infty,P}$ can be practically rather difficult.
For an explicit algorithm, \cite{Joseph2015} proposed a greedy method, which (for the case $\delta \rightarrow \infty$) is defined as
\begin{eqnarray}
x_1 \; \in \; \argmax_{x \in \mathcal{X}} \; \tilde{p}(x), \qquad x_n \; \in \; \argmax_{x \in \mathcal{X}} \; \min_{i = 1,\dots,n-1} \left[ \frac{1}{2d} \log \tilde{p}(x_i) + \frac{1}{2d} \log \tilde{p}(x) + \log \|x_i - x\|_2 \right] \quad n \geq 2 . \label{ap: med algo}
\end{eqnarray}
The method was recently implemented in the \verb+R+ package \verb+mined+, available at \url{https://cran.r-project.org/web/packages/mined/}.
Although the results presented in this work are based on our own implementation of \eqref{ap: med algo} in \verb+MATLAB+, it would be interesting in further work to explore the extent to which the performance of MED is improved in \verb+mined+.

For the results reported in this paper, the global optimisation in \eqref{ap: med algo} was replaced by an adaptive Monte Carlo optimisation method.
Indeed, in Fig. 2(c) of \cite{Chen2018SteinPoints} it was established that MED performed best for SP when an adaptive Monte Carlo optimisation was used, compared to both a basic grid search and the Nelder--Mead method \cite{Nelder1965}.
Specifically, the adaptive Monte Carlo method is described in Alg. \ref{adaptive MC search}.
Here $\mu_0$ and $\Sigma_0$ are the mean and covariance of a Gaussian and are selected to approximately match the first two moments of $P$.
The notation $\Pi(\{z_i\}_{i=1}^n,\lambda)$ denotes a uniform-weighted Gaussian mixture distribution with each component having identical variance $\lambda$, to be specified.
Thus the algorithm described in Alg. \ref{adaptive MC search} picks randomly between drawing a set $\{x_i^{\text{test}}\}_{i=1}^{n_{\text{test}}}$ from $\mathcal{N}(\mu_0,\Sigma_0)$ (with probability $\alpha_n$ to be specified) and drawing such a set instead from $\Pi(\{x_i\}_{i=1}^{n-1}, \lambda)$, a Gaussian mixture based on the current point set $\{x_i\}_{i=1}^{n-1}$.

For our experiments, the parameters $\alpha_n,\mu,\Sigma_0,n_{\text{test}}$ of Alg. \ref{adaptive MC search} were approximately optimised in favour of MED. 
However, it is likely that the numerical optimisation routine used in \verb+mined+ will lead to different results to those reported in this work, and it may be possible to obtain better performance when more sophisticated numerical optimisation routines are used.

\begin{algorithm}
\caption{Adaptive Monte Carlo search method to select $x_n$}
\label{adaptive MC search}
\begin{algorithmic}[1]
\STATE Draw $u \sim \mathrm{Unif(0,1)}$ 
\IF{$u \le \alpha_n$}
\STATE Draw $\{x^{\mathrm{test}}_{i}\}_{i=1}^{n_{\mathrm{test}}} \sim \mathcal{N}(\mu_0, \Sigma_0)$
\ELSE
\STATE Draw $\{x^{\mathrm{test}}_{i}\}_{i=1}^{n_{\mathrm{test}}} \sim \Pi(\{x_i\}_{i=1}^{n-1}, \lambda)$
\ENDIF
\STATE $j^* \gets \argmin_{j \in\{1 \ldots n_{\mathrm{test}}\}} \log \mathcal{E}_{\infty,P} (\{x_i\}_{i=1}^{n-1} \cup \{x^{\mathrm{test}}_{j} \} )$
\STATE $x_n \gets x^{\mathrm{test}}_{j^*}$
\end{algorithmic}
\end{algorithm}

\paragraph{Stein Variational Gradient Descent} 

The SVGD method was first proposed in \cite{Liu2016SVGD} and subsequently studied in (e.g.) \cite{Liu2017GradientFlow,Liu2018}.
The approach is rooted in a continuous version of gradient descent on $\mathcal{P}(\mathcal{X})$, the set of probability distributions on $\mathcal{X}$, with the Kullback-Leibler divergence $\text{KL}(\cdot || P)$ providing the gradient.
To this end, restrict attention to $\mathcal{X} = \mathbb{R}^d$, let $\mathcal{K}$ be a RKHS as in the main text and consider the discrete time process
\begin{eqnarray*}
S_f(x) & = & x + \epsilon g(x)
\end{eqnarray*}
parametrised by a function $g \in \mathcal{K}^d$.
For an infinitesimal time step $\epsilon$ we can lift $S_g$ to a pushforward map on $\mathcal{P}(\mathcal{X})$; i.e. $Q \mapsto (S_g)_{\#} Q$.
Then \cite{Liu2016SVGD} established that
\begin{eqnarray}
- \left. \frac{\mathrm{d}}{\mathrm{d} \epsilon} \mathrm{KL}((S_g)_{\#} Q || P) \right|_{\epsilon = 0} & = & \int_{\mathcal{X}} \mathcal{A} g \; \mathrm{d}Q  \label{eq: KL diff}
\end{eqnarray}
where $\mathcal{A}$ is the Langevin Stein operator defined in the main text; i.e. $\mathcal{A}g = \frac{1}{\tilde{p}} \nabla \cdot (\tilde{p} g)$.
The direction of fastest descent
\begin{eqnarray*}
g^*(\cdot) & := & \argmax_{g \in B(\mathcal{K}^d)} \; - \left. \frac{\mathrm{d}}{\mathrm{d} \epsilon} \mathrm{KL}((S_g)_{\#} Q || P) \right|_{\epsilon = 0}
\end{eqnarray*}
has a closed form with $j$th coordinate equal to (in informal notation)
\begin{eqnarray*}
g_j^*(\cdot ; Q) & = & \int_{\mathcal{X}} (\partial_{x^j} + \partial_{x^j} \log \tilde{p}) k(x,\cdot) \; \mathrm{d}Q(x) .
\end{eqnarray*}
To obtain a practical algorithm, \cite{Liu2016SVGD} proposed to discretise this dynamics in both space $\mathcal{X}$, through the use of an empirical approximation to $Q$, and in time, through the use of a fixed and positive time step $\epsilon > 0$.
The result is a sequence of empirical measures based on point sets $\{x_i^{(m)}\}_{i=1}^n$ for $m \in \mathbb{N}$, where in what follows we have re-purposed superscripts to denote iteration number instead of coordinate.
Thus, given an initialisation $\{x_i^{(0)}\}_{i=1}^n$ of the points, at iteration $m \geq 1$ of the algorithm we update 
\begin{eqnarray*}
x_i^{(m)} \; = \; x_i^{(m-1)} + \epsilon g^*(x_i^{(m-1)} ; Q_n^m) , \qquad Q_n^m := \frac{1}{n} \sum_{i=1}^n \delta_{x_i^{(m-1)}}
\end{eqnarray*}
in parallel at a computational cost of $O(n)$.
The output is the empirical measure $Q_n^m$ and positive theoretical results on the convergence of $Q_n^m$ to $P$ is at present an open research question, though continuum versions of SVGD have now been studied (e.g.) \cite{Liu2017GradientFlow,Lu2018}.
In addition, recent work has sought to improve the empirical performance of SVGD by the use of quasi-Newton methods; see \cite{Detommaso2018}.
However, for all experiments in this paper we employed the original formulation of SVGD due to \cite{Liu2016SVGD}.

\subsubsection{Gaussian Mixture Model}


SP was implemented based on the same adaptive Monte Carlo search procedure described above in Alg. \ref{adaptive MC search}.
To ensure a fair comparison, the number of search points was taken equal to $m_j$, the length of the Markov chain used in SP-MCMC.

\subsubsection{IGARCH Model} 
\label{subsubsec: IGARCH implementation}

SP and MED were each implemented based on the same adaptive Monte Carlo search procedure described above in Alg. \ref{adaptive MC search}.
To ensure a fair comparison, in each case the number of search points was taken equal to $m_j$, the length of the Markov chain used in SP-MCMC.

For SVGD the size $n$ of the point set must be pre-specified.
In order to ensure a fair comparison with SP-MCMC we considered a point set of size $n = 1000$, which is identical to the size of point set that are ultimately produced by SP-MCMC during the course of this experiment.
The step size $\epsilon$ was hand-tuned to optimise the performance of SVGD in our experiment.
The point set was initialised for SVGD by sampling each point independently from $\mathrm{Unif}((0.002, 0.04) \times (0.05, 0.2))$.
The step-size $\epsilon$ for SVGD was set using Adagrad, as in \cite{Liu2016SVGD}, with \emph{master step size} $0.001$ and \emph{momentum} $0.9$.

The resulting point sets are visualised in Fig. \ref{fig: IGARCH points}.

\begin{figure}
\centering
\includegraphics[width = 0.9\textwidth]{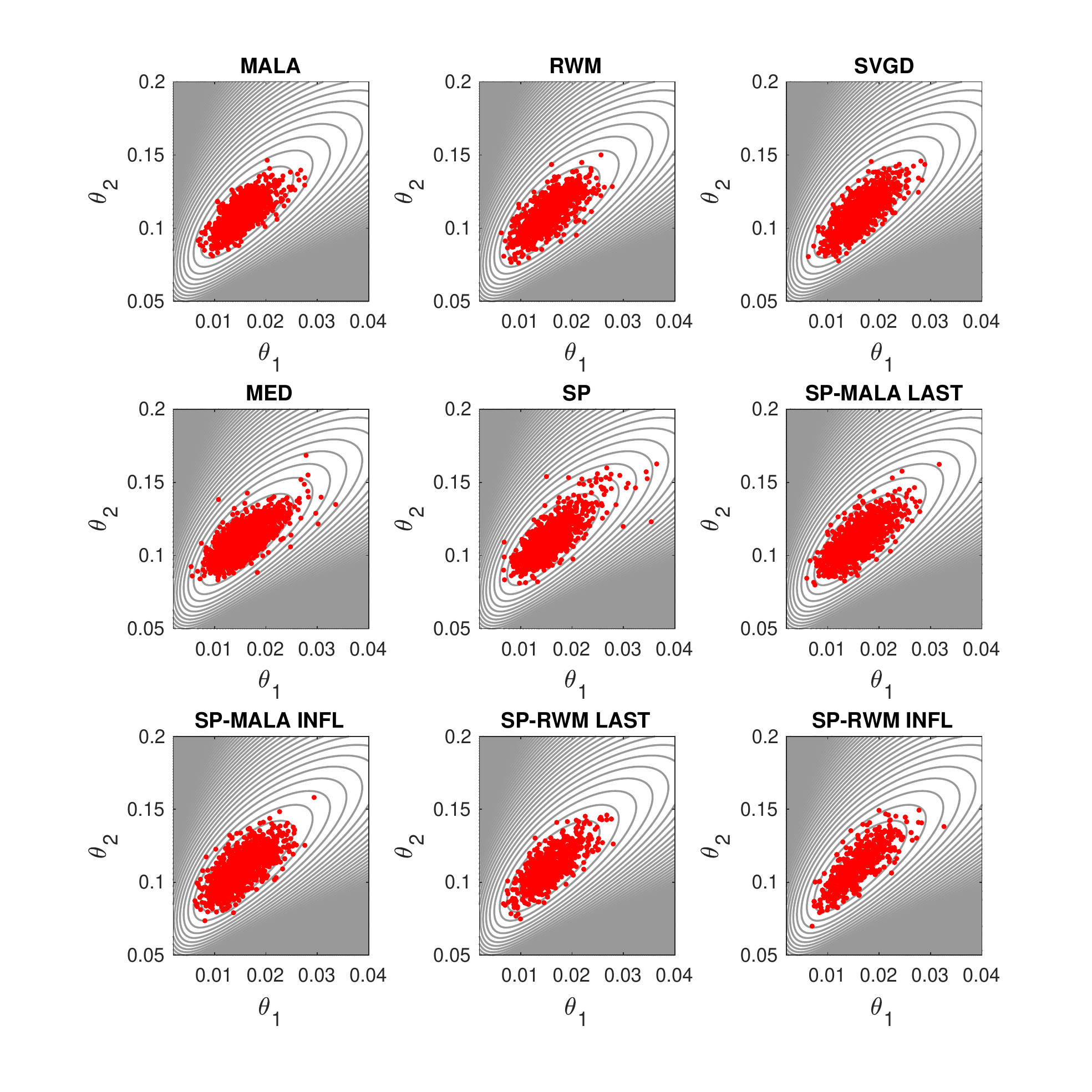}
\caption{Quantisation in the IGARCH experiment.
For the IGARCH experiment, we display point sets of size $n = 1000$ produced by each of MCMC, SP, MED, SVGD and SP-MCMC, as described in Section \ref{subsec: IGARCH} and Appendix \ref{subsubsec: IGARCH implementation}.}
\label{fig: IGARCH points}
\end{figure}

\subsubsection{Performance of the Preconditioner Kernel} \label{subsubsec: preconditioner test}

The performance of the preconditioner kernel in \eqref{eq: IMQ kernel} was explored by comparing the default case $\Lambda = I$ (i.e. where a preconditioner is not used) to the case where $\Lambda$ is estimated based on a short run of MCMC.
For this comparison we first considered the IGARCH experiment described in Sec. \ref{subsec: IGARCH}.
This is expected to prove to be a challenging example for a preconditioner, in the sense the posterior is already unimodal and fairly well-conditioned.
To test this hypothesis, the experiment reported in the main text was performed for both selections of $\Lambda$ and the results are reported in Fig. \ref{fig: preconditioner test} (left).
It can indeed be seen that the use of a preconditioner does not lead to a gain in performance in this case; however, and importantly, performance does not get \emph{worse} as a result of the preconditioner being used. 

\begin{figure}
\centering
\begin{subfigure}[t]{0.49\textwidth}
\centering
\includegraphics[width = \textwidth]{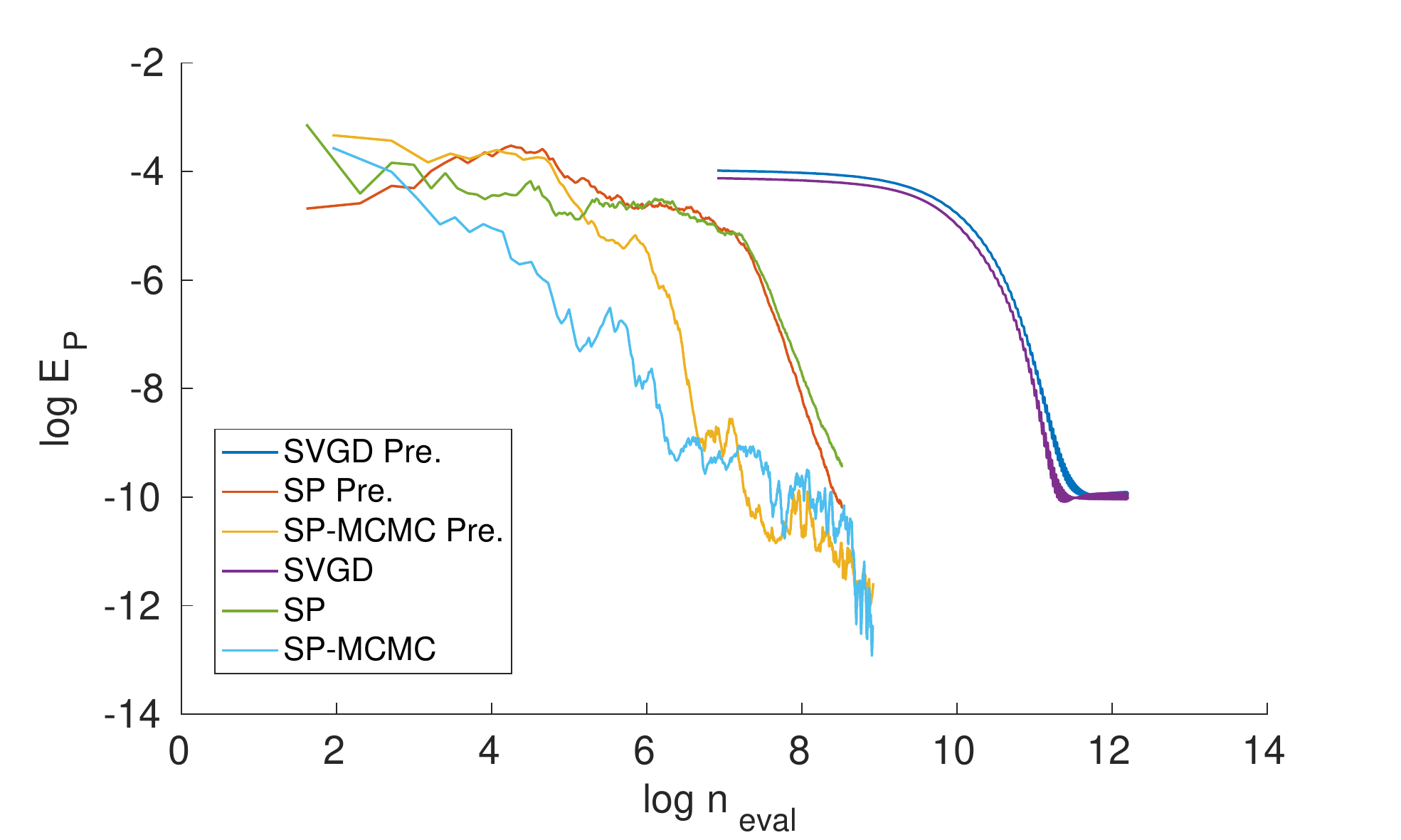}
\caption{IGARCH Model}
\end{subfigure}
\begin{subfigure}[t]{0.49\textwidth}
\centering
\includegraphics[width = \textwidth]{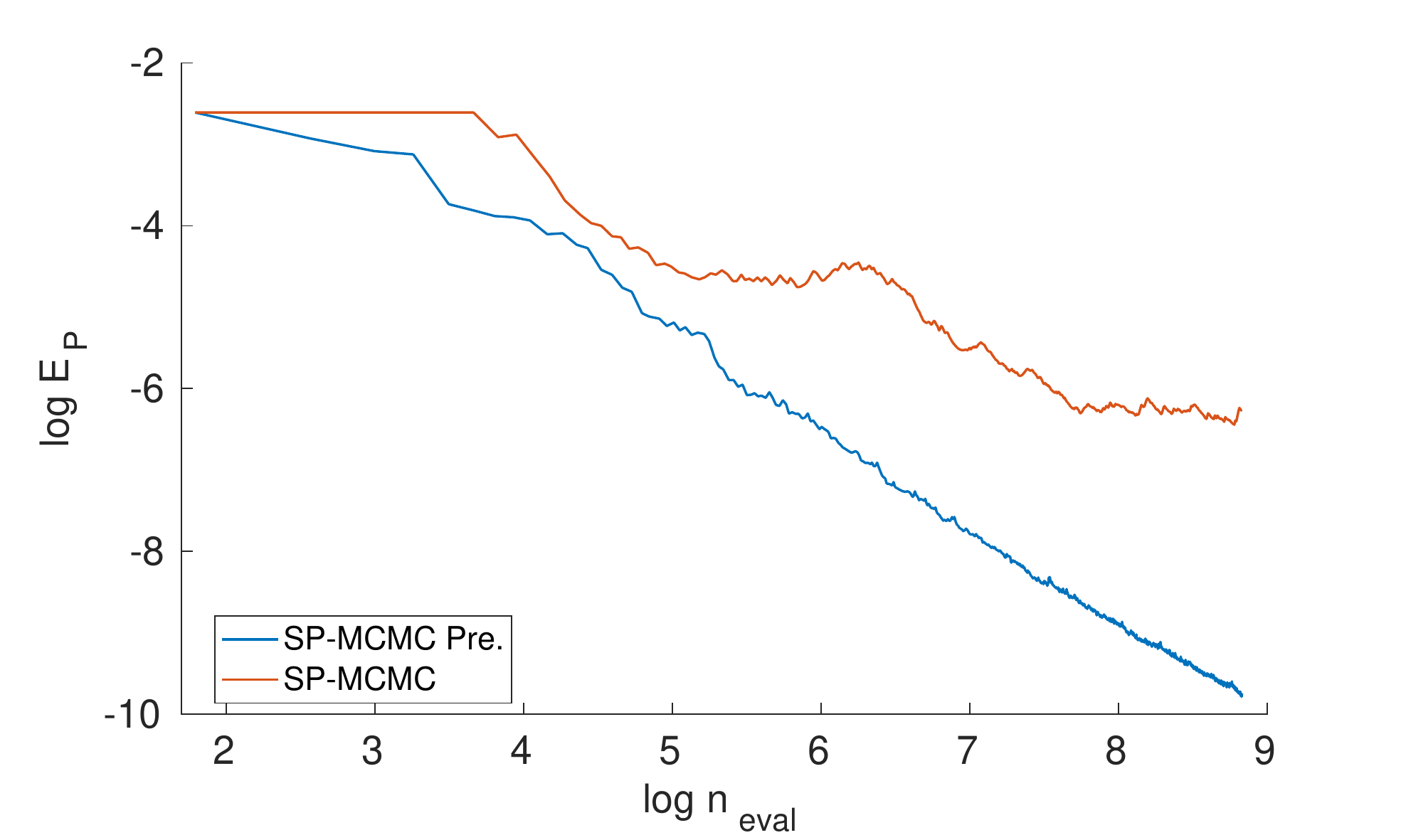}
\caption{Toy Model}
\end{subfigure}
\caption{Assessing the contribution of the preconditioner kernel.
For the IGARCH experiment (left), we compared the performance of SP, SVGD and SP-MCMC each with and without a preconditioner matrix $\Lambda$ being used.
In addition, a toy model (right) was explored for which the natural scale of the state vector $x$ was not commensurate with the naive choice of preconditioner matrix $\Lambda = I$.
As in Fig. \ref{fig:IGARCH} in the main text, each method produced an empirical measure $\frac{1}{n}\sum_{i=1}^n \delta_{x_i}$ whose distance to the target $P$ was quantified by the energy distance $E_P$.
The computational cost was quantified by the number $n_{\text{eval}}$ of times either $\tilde{p}$ or its gradient were evaluated.}
\label{fig: preconditioner test}
\end{figure}

To explore a scenario where the preconditioner demonstrates a benefit, we considered a toy model $P = \mathcal{N}(0,\sigma^2 I)$ with $\sigma = 0.01$.
In this case the naive choice of $\Lambda = I$ was compared to the proposed approach of taking $\Lambda$ to be a sample-based approximation to the covariance of $P$.
It is seen in Fig. \ref{fig: preconditioner test} (right) that the preconditioner kernel out-performed the naive choice of $\Lambda = I$ for this model, emphasising the need to ensure $\Lambda$ is commensurate with the scale of the state variable $x$ in general.

An important point is that there is in general a computational cost associated with building an appropriate preconditioner matrix $\Lambda$ as just described.
This is not explicitly accounted for in the experimental results that we report (i.e. not included in the total $n_{\text{eval}}$).
To address this point, and to demonstrate that SP-MCMC can prove to be effective in the absence of this computational overhead, we refer the reader to the ODE example in Secs. \ref{subsec: ODE} and \ref{subsubsec: Goodwin details}, where a simple preconditioner matrix $\Lambda \propto I$ was used and strong results were nevertheless obtained.

\subsubsection{Removal of ``Bad'' Points} \label{subsec: bad points supp}

In this section we explored the implications of Remark \ref{rem: bad points} in the main text, which proposed to remove ``bad'' points from the current point set.
This can be motivated by analogous strategies, known as \emph{away steps}, explored in the Frank-Wolfe literature \citep[e.g.][]{LacosteJulien2015,Freund2017}.
An approach was considered such that, if the current point set in SP-MCMC is $\{x_i\}_{i=1}^{j-1}$, then in addition to identifying a possible next point $x_j$, we also identify a ``bad'' point $x_{i^*}$ that minimises $D_{\mathcal{K}_0,P}(\{x_i\}_{i=1}^{j-1} \setminus \{x_{i^*}\})$.
Then we compare the two quantities
\begin{eqnarray*}
\Delta_{\text{good}} & := & D_{\mathcal{K}_0,P}(\{x_i\}_{i=1}^{j-1}) - D_{\mathcal{K}_0,P}(\{x_i\}_{i=1}^j) \\
\Delta_{\text{bad}} & := & D_{\mathcal{K}_0,P}(\{x_i\}_{i=1}^{j-1}) - D_{\mathcal{K}_0,P}(\{x_i\}_{i=1}^{j-1} \setminus \{x_{i^*}\}) .
\end{eqnarray*}
If either $j=1$ or $\Delta_{\text{good}} > \Delta_{\text{bad}}$, then the new point $x_j$ is included in the point set, so that the updated point set is $\{x_i\}_{i=1}^j$.
Otherwise we remove the ``bad'' point from the point set, so that the update point set is $\{x_i\}_{i=1}^{j-1} \setminus \{x_{i^*}\}$.
As such, with this approach the iteration number of the SP-MCMC algorithm is no longer identical to the size of the point set and the computational cost required to obtain a size $n$ point set may be increased relative to vanilla SP-MCMC.
However, it is possible that this approach can lead to improved performance at the quantisation task.

To empirically test this approach, we revisit the IGARCH experiment in Sec. \ref{subsec: IGARCH} of the main text.
The SP-MCMC method was implemented with the \verb!INFL! criterion and with $m_j = 5$, identical to the experiment shown in the main text.
Results comparing the impact of removing ``bad points'' in the manner described above are shown in Fig. \ref{fig: drop test}.
Interestingly, this was not seen to work well because too frequently a point would be removed, leading to sets containing only a small number of points.
To investigate further, we considered an alternative approach wherein the ``current worst'' point would occasionally be removed.
Results are also depicted in Fig. \ref{fig: drop test}, based on a ``dropping'' rate of 25\%.
This led to larger point sets and an improved performance as measured by $E_P$.
However, neither strategy considered outperformed the default approach where points were never removed.

\begin{figure}[t!]
\centering
\includegraphics[width = 0.5\textwidth]{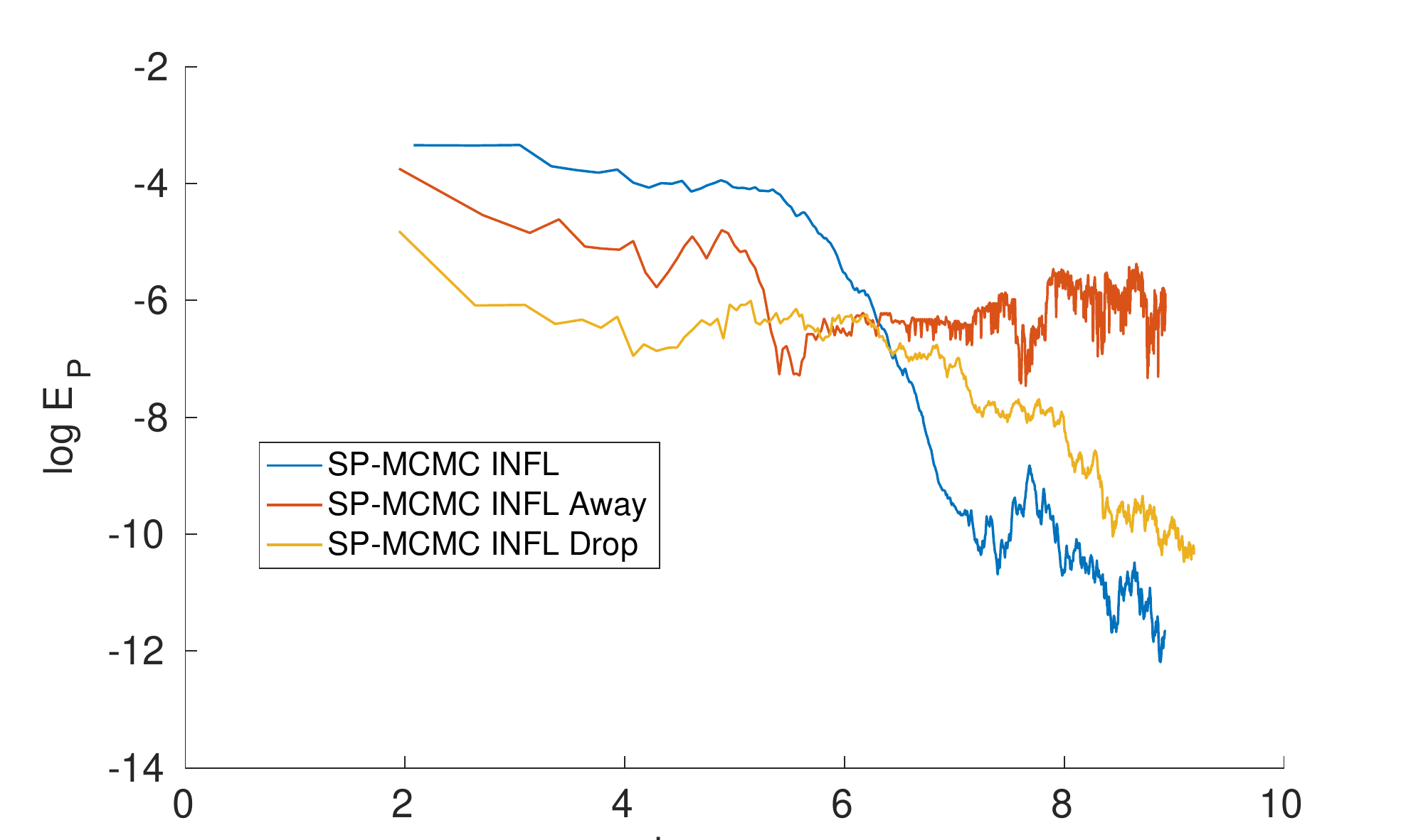}
\caption{Assessing the benefit of removing ``bad'' points.
For the IGARCH experiment, we compared the performance of SP-MCMC with and without \emph{away steps} (``Away'') being used.
In addition a simpler procedure, whereby the ``current worst'' point is occasionally dropped (``Drop''), was considered.
As in Fig. \ref{fig:IGARCH} in the main text, each method produced an empirical measure $\frac{1}{n}\sum_{i=1}^n \delta_{x_i}$ whose distance to the target $P$ was quantified by the energy distance $E_P$.
The computational cost was quantified by the number $n_{\text{eval}}$ of times either $\tilde{p}$ or its gradient were evaluated.}
\label{fig: drop test}
\end{figure}

\subsubsection{Goodwin Oscillator} \label{subsubsec: Goodwin details}

The \emph{Goodwin oscillator} \citep{Goodwin1965a} is a dynamical model of oscillatory enzymatic control.
This kinetic model, specified by a system of $q$ ODEs, describes how a negative feedback loop between protein expression and mRNA transcription can induce oscillatory dynamics at a cellular level.
In this work we considered Bayesian parameter estimation for two such models; a simple model with no intermediate protein species ($q=2$) and a more complex model with six intermediate protein species ($q=8$).
The experimental protocol below follows that used in the earlier work of \cite{Calderhead2009,Oates2016}.

The Goodwin oscillator with $q$ species is given by
\begin{eqnarray}
\frac{\mathrm{d}u_1}{dt} & = & \frac{a_1}{1+a_2u_q^\rho} - \alpha u_1 \label{Goodwin} \\
\frac{\mathrm{d}u_2}{dt} & = & k_1u_1 - \alpha u_2 \nonumber \\
& \vdots & \nonumber \\
\frac{\mathrm{d}u_q}{dt} & = & k_{q-1}u_{q-1} - \alpha u_q. \nonumber
\end{eqnarray}
The variable $u_1$ represents the concentration of mRNA and $u_2$ represents its corresponding protein product.
The variables $u_3,\dots,u_q$ represent intermediate protein species that facilitate a cascade of enzymatic activation leading, ultimately, to a negative feedback, via $u_q$, on the rate at which mRNA is transcribed.
The Goodwin oscillator permits oscillatory solutions only when $\rho > 8$.
Following \cite{Calderhead2009,Oates2016} we set $\rho = 10$ as a fixed parameter.
The solution $u(t)$ of this dynamical system depends upon synthesis rate constants $a_1$, $k_1,\dots,k_{q-1}$ and degradation rate constants $a_2$, $\alpha$.
Thus the parameter vector $\theta = [a_1,a_2,k_1,\dots,k_{q-1},\alpha] \in [0,\infty)^{q+2}$ and a $q$-variable Goodwin model has $d = q+2$ uncertain parameters to be inferred. 

For this experiment we followed \cite{Oates2016} and considered a realistic setting where only mRNA and protein product are observed, corresponding to $g(u) = [u_1,u_2]$. 
For the initial condition we took $u_0 = [0,\dots,0]$ and for the measurement noise we took $\sigma = 0.1$, both considered known and fixed.
Data $\{y_i\}_{i=1}^{40}$ were generated using $a_1 = 1$, $a_2 = 3$, $k_1 = 2$, $k_2,\dots,k_{q-1} = 1$, $\alpha = 0.5$, as in \cite{Oates2016}.
Thus the likelihood function for these data has the Gaussian form
\begin{eqnarray*}
p(y | \theta) & = & \frac{1}{(2 \pi \sigma^2)^{20}} \exp\left( - \frac{1}{2\sigma^2} \sum_{i=1}^{40} \|y_i - g(u(t_i))\|_2^2 \right) .
\end{eqnarray*}
To set up the Bayesian inferential framework, each parameter $\theta_i$ was assigned an independent $\Gamma(2,1)$ prior belief. 
Note that, in order to ensure that boundary conditions in Sec. \ref{subsec: discrepancy} were satisfied, we subsequently worked with the log-transformed parameters $x = \log (\theta) \in \mathbb{R}^d$, identifying the posterior distribution of $x$ with the target $P$ in our assessment.

In order to obtain the gradient $\nabla \log \tilde{p}$ it is required to differentiate the solution $u(t_i)$ of the ODE with respect to the parameters $x$ at each time point $i \in \{1,\dots,40\}$.
Of course, from the chain rule it is sufficient in what follows to consider differentiation of $u(t_i)$ with respect to $\theta$.
To this end, define the \emph{sensitivities} $S_{j,k}^i := \frac{\partial u_k}{\partial \theta_i}(t_j)$, and note that these satisfy 
\begin{eqnarray}
\frac{\mathrm{d}}{\mathrm{d}t} S_{j,k}^i = \frac{\partial f_k}{\partial \theta_i} + \sum_{l=1}^q \frac{\partial f_k}{\partial u_l} S_{j,l}^i \label{sense ode}
\end{eqnarray}
where $\frac{\partial u_k}{\partial\theta_i} = 0$ at $t = 0$.
The sensitivities can therefore be numerically computed by augmenting the state vector $u$ of the original ODE to include the $S_{j,k}^i$.
In this work the \verb!ode45! solver in \verb!MATLAB! was used to numerically solve this augmented ODE.

For SP-MCMC, we found that construction of a suitable preconditioner matrix $\Lambda$ was challenging due to the computational cost associated with each forward-solve of the ODE.
To this end, we simply selected $\Lambda = 0.3I$ for the case $d = 4$ and $\Lambda = 0.15 I$ for the case $d = 10$.
The same kernel $k$, with this choice of $\Lambda$, was employed for each of SP, SP-MCMC and SVGD.

SP and MED were each implemented based on the same adaptive Monte Carlo search procedure described above in Alg. \ref{adaptive MC search}.
To ensure a fair comparison, in each case the number of search points was taken equal to $m_j$, the length of the Markov chain used in SP-MCMC.
The methods were run to produce a point set of size $n = 300$.

For SVGD the size $n$ of the point set must be pre-specified.
In order to ensure a fair comparison with SP-MCMC we considered a point set of size $n = 300$, which is identical to the size of point set that are untimately produced by SP-MCMC during the course of this experiment.
The step size $\epsilon$ was hand-tuned to optimise the performance of SVGD in our experiment.
To initialise SVGD, a point set was drawn independently at random from $\text{Uniform}(\theta^* - 0.5 \times 1, \theta^* + 0.5 \times 1)$, where $1 = [1,\dots,1]$ where $\theta^*$ is the data-generating value of the parameter vector.
The step-size $\epsilon$ for SVGD was set using Adagrad, as in \cite{Liu2016SVGD}, with \emph{master step size} $0.005$ and \emph{momentum} $0.9$.


\end{document}